\documentclass[journal]{IEEEtran}
\usepackage[utf8]{inputenc}

\usepackage{array}
\usepackage{tabularx}
\usepackage{subcaption}
\usepackage{amssymb,amsfonts,amsmath}
\usepackage{amsthm}
\usepackage{mathtools}
\usepackage{bm}
\usepackage{hyperref}
\usepackage{enumitem}
\usepackage{graphicx}
\usepackage{booktabs}
\usepackage{multirow}
\usepackage{float}
\usepackage{algorithm}
\usepackage{algpseudocode}
\usepackage{wasysym}
\usepackage{xcolor}
\usepackage{pifont}
\usepackage{cleveref}
\usepackage{placeins}
\usepackage[automake]{glossaries}
\makeglossaries
\usepackage{etoolbox}
\usepackage{etoolbox}
\newtheorem{proposition}{Proposition}[section]

\usepackage{makecell}

\newtoggle{printappendix}
\toggletrue{printappendix}
\usepackage[automake]{glossaries}
\makeglossaries
\glsdisablehyper

\newglossaryentry{postcuttranscript}{
  name={post-cut transcript},
  description={the full collection of observable metadata artefacts
  generated when a quantum circuit is decomposed into fragments and
  submitted to a cloud provider for execution}
}

\newglossaryentry{wsi}{
  name={workload-structure inference},
  description={the act of recovering algorithmic or structural
  properties of a hidden quantum workload from observable
  submission-time metadata alone, without access to quantum
  states or measurement outcomes}
}

\newglossaryentry{routingtax}{
  name={topological transpilation penalty},
  description={the topology-dependent compilation overhead incurred
  when mapping a logical circuit onto constrained physical hardware,
  acting as the primary physical mechanism through which structural
  information about the hidden workload is encoded into
  provider-visible metadata}
}

\newglossaryentry{transcriptthreatmodel}{
  name={transcript threat model},
  description={the formalised confidentiality surface defined by
  the cloud-visible post-cut execution transcript, within which
  a semi-honest adversary performs passive offline inference}
}

\newglossaryentry{postcutinference}{
  name={post-cut metadata inference},
  description={the attack surface defined by the observable
  fragment-level execution transcript, through which a semi-honest
  cloud provider can reconstruct computational intent from compiled
  circuit metadata without accessing quantum states or measurement
  outcomes}
}

\newglossaryentry{submissiontranscript}{
  name={submission-time transcript},
  description={the subset of the observable transcript available
  at the submission interface prior to execution, comprising
  compiled fragment width, depth, two-qubit gate count, shot
  allocation, and submission timestamp; formally denoted
  $\mathcal{L}$}
}

\newglossaryentry{executionartefact}{
  name={execution-time artefact},
  description={the subset of the observable transcript available
  post-dispatch during cloud ingest and orchestration, comprising
  execution order and true QPU execution duration; formally
  denoted $\Pi$}
}

\newglossaryentry{matchedfootprint}{
  name={matched-footprint control},
  description={an evaluation method that restricts comparison to
  similarly sized circuit sub-populations using a caliper in
  normalised width, depth, and fragment count space, removing
  coarse scale cues to test whether leakage survives footprint
  equalisation}
}

\newglossaryentry{instancedisjoint}{
  name={instance-disjoint split},
  description={an evaluation protocol in which no cutting job
  appears in both training and test sets, preventing
  memorisation of specific circuit instances}
}

\newglossaryentry{sizeholdout}{
  name={size-holdout split},
  description={an evaluation protocol in which the largest-width
  quartile of the dataset is held out as the test set, measuring
  generalisation to unseen workload scales}
}

\newglossaryentry{structuralfingerprint}{
  name={structural fingerprint},
  description={the statistical distribution of compiled sub-circuit
  metadata that encodes the topological and algorithmic properties
  of the original monolithic workload}
}

\newglossaryentry{semihonest}{
  name={semi-honest adversary},
  description={a cloud provider that executes fragment jobs
  correctly but passively logs and analyses classical metadata
  available at the submission interface, deriving advantage
  solely from offline statistical inference; also referred to
  as honest-but-curious}
}

\newglossaryentry{fragmentmultiplicity}{
  name={fragment multiplicity},
  description={the number of independent sub-circuit jobs
  arising from a single cut workload, forming part of the
  observable execution transcript}
}

\newglossaryentry{depthinflation}{
  name={depth inflation},
  description={the ratio of compiled circuit depth to logical
  circuit depth, quantifying the routing overhead imposed by
  mapping a logical circuit onto a constrained hardware topology;
  the dominant leakage signal carrier across all inference tasks}
}

\newglossaryentry{compiledwidth}{
  name={compiled width},
  description={the logical qubit count of a circuit fragment after
  partitioning, acting as a topology-invariant structural signal
  determined by partitioning strategy and algorithm structure
  rather than hardware routing}
}

\newglossaryentry{postcutleakage}{
  name={post-cut metadata leakage},
  description={the confidentiality risk arising from fragmented
  quantum execution, wherein the resulting observable metadata
  forms a structural footprint that a cloud provider can exploit
  to infer properties of the hidden workload}
}

\newglossaryentry{circuitcutting}{
  name={quantum circuit cutting},
  description={the process of decomposing a monolithic quantum
  circuit into smaller, independently executable fragments,
  effectively expanding the computational capacity of
  noisy intermediate-scale quantum (NISQ) devices}
}

\newglossaryentry{twoqoverhead}{
  name={two-qubit routing overhead},
  description={the absolute number of extra two-qubit gates
  added to a circuit during compilation to satisfy the
  connectivity constraints of a specific hardware topology}
}

\newglossaryentry{observabletranscript}{
  name={observable transcript},
  description={the complete metadata surface exposed to the cloud
  provider, formally defined as the union of the submission-time
  transcript $\mathcal{L}$ and the execution-time artefact $\Pi$;
  denoted $\mathcal{T} = \mathcal{L} \cup \Pi$}
}

\newglossaryentry{trustedclient}{
  name={trusted client domain},
  description={the local execution environment in which the
  original unpartitioned circuit, the partitioning decisions,
  and the stitching specification are generated and kept secret
  from the cloud provider}
}

\newglossaryentry{untrustedcloud}{
  name={untrusted cloud service},
  description={the semi-honest provider that manages job
  ingestion, scheduling, and backend execution, executing
  fragment jobs correctly while passively logging classical
  metadata to perform offline statistical inference}
}

\newglossaryentry{wsiind}{
  name={WSI-IND},
  description={the Workload-Structure Inference indistinguishability
  game formalising the security threat, wherein a semi-honest
  adversary attempts to guess the underlying algorithm family
  from the observable compiled metadata transcript}
}

\newglossaryentry{inferenceobjectives}{
  name={inference objective},
  description={the six classification targets pursued by the
  adversary: cutting mechanism (W1), algorithm family (A1),
  sub-family variant (A2), connectivity regime (H1), geometry
  class (H2), and $k$-locality regime (H3). A1 targets eight
  core families (HEA, QAOA, QFT, Random, QML, Sim, Chem,
  Oracle); A2 targets their structural sub-variants such as
  the standard, no-swaps, and approximate variants of QFT
  or the linear, circular, and reverse-linear structures
  of HEA}
}

\usepackage{acronym}

\newacro{API}{Application Programming Interface}
\newacro{AUC}{Area Under the Receiver Operating Characteristic Curve}
\newacro{BQC}{Blind Quantum Computation}
\newacro{HEA}{Hardware-Efficient Ansätz}
\newacro{NISQ}{Noisy Intermediate-Scale Quantum}
\newacro{QAOA}{Quantum Approximate Optimization Algorithm}
\newacro{QFT}{Quantum Fourier Transform}
\newacro{QHE}{Quantum Homomorphic Encryption}
\newacro{QML}{Quantum Machine Learning}
\newacro{QPU}{Quantum Processing Unit}
\newacro{UBQC}{Universal Blind Quantum Computation}

\newacro{OvR}{One-vs-Rest}
\usepackage{svg}

\definecolor{FullGreen}{RGB}{34,139,34}
\definecolor{PartialYellow}{RGB}{255,179,0}

\newcommand{\full}{\textcolor{FullGreen}{\ding{108}}}
\newcommand{\pfull}{\textcolor{PartialYellow}{\ding{108}}}
\newcommand{\none}{\textcolor{black!30}{\ding{109}}}

\newcommand{\Prob}{\mathbb{P}}

\newcommand{\transcript}{\mathcal{T}}
\newcommand{\secret}{\mathcal{S}}
\newcommand{\observable}{\mathcal{L}}
\newcommand{\artefacts}{\Pi}

\numberwithin{equation}{section}
\theoremstyle{plain}

\title{Post-Cut Metadata Inference Attacks on\\
       Quantum Circuit Cutting Pipelines}

\author{
    \IEEEauthorblockN{Samuel Punch, Krishnendu Guha, and Utz Roedig}\\
    \IEEEauthorblockA{School of Computer Science and Information Technology\\
    University College Cork, Ireland\\
    \{samuel.punch, kguha, u.roedig\}@ucc.ie}
}
       
\begin{document}
\maketitle

\begin{abstract}
Quantum cloud providers can identify a user's algorithm and secret problem 
structure without ever seeing actual quantum data, simply by analyzing routine 
metadata collected for billing and system management. Existing confidentiality 
tools such as blind quantum computation and quantum homomorphic encryption 
protect the quantum payload itself, but they do not protect this classical 
orchestration metadata. This leaves an unexplored security risk in the logs 
generated when a large quantum program is split into smaller pieces to fit 
onto limited hardware, a process known as \emph{circuit cutting}.

These fragments leak sensitive information through what we term the 
\textbf{topological transpilation penalty}: the unavoidable depth and gate 
inflation added when a compiler reorganizes a program for a restricted 
hardware topology. Tests on a 156-qubit production \textbf{Quantum Processing 
Unit (QPU)} show that traditional timing side-channels fail in this setting, 
since hardware control-plane delays mask actual quantum execution time. The 
unique shape of the transpilation penalty acts instead as a persistent 
structural fingerprint for the hidden workload.

Using 12,000 circuit fragments across eight algorithm families, our attack 
recovers algorithm family and Hamiltonian $k$-locality with near-perfect 
accuracy, achieving instance-disjoint AUC $= 1.000$ for both. This leakage 
persists under size-holdout evaluation on unseen circuit scales, with AUC 
$= 0.987$ and $0.986$ respectively. The cutting mechanism is inferred with 
AUC $= 0.991$, and hardware topology is recovered well above chance with 
AUC $= 0.818$. These results show that circuit cutting exposes algorithmic 
intent, and potentially proprietary problem structure, through metadata 
alone, without any need to observe quantum data.
\end{abstract}


\section{Introduction}
As quantum algorithms grow in complexity, they often exceed the physical 
qubit counts and connectivity limits of today's \ac{NISQ} devices. To 
bridge this gap, \emph{circuit cutting} (CC) has become an essential 
technique for executing large circuits by partitioning them into smaller 
fragments that fit on available hardware \cite{Peng_2020, Tang_2021}. 
This problem matters now more than ever as circuit cutting becomes the 
standard for running ``utility-scale'' experiments on hardware with 
100+ qubits. Tools like Blind Quantum Computation (BQC) 
\cite{broadbent2009ubqc} and Quantum Homomorphic Encryption (QHE) 
\cite{Dulek_2018} were designed to hide the quantum data itself: the 
gate sequence, the measurement outcomes, the payload the provider 
executes. Neither protects the classical orchestration metadata 
generated when a job is split into fragments, including fragment 
count, per-fragment size, and submission order. This is the gap that 
existing confidentiality frameworks leave open, and it is the gap this 
paper closes.

While circuit cutting expands the reach of current quantum computers, 
it introduces a critical security risk: the leakage of workload 
metadata. Transitioning from a monolithic job to distributed fragments fundamentally alters the provider's observability window. This process generates a 
detailed \textbf{post-cut metadata transcript}, a collection of 
administrative records (such as fragment width, depth, and gate counts) 
used for billing and system management 
\cite{aws_braket_cloudwatch, ibm_runtime_workload}. Even if the provider 
never sees the actual quantum data, this ``packaging'' data acts as a 
powerful side-channel \cite{9951250}. A curious or ``semi-honest'' 
provider can analyse these logs to perform \emph{workload-structure 
inference}, reconstructing the user's secret algorithm and problem 
structure.

Previous research has focused on ``timing side-channels'' 
\cite{lu2024quantumleaktimingsidechannel, dong2025exploitingtimingsidechannelsquantum}, 
but our validation on a \textbf{156-qubit production quantum computer} 
reveals that execution time stays almost unchanged even when the 
compiled circuit becomes up to \textbf{$23\times$} deeper. The reason is 
that the quantum operations themselves are so brief that they are 
overwhelmed by the fixed classical control and scheduling delays 
surrounding every job, which leaves traditional timing-based observers 
effectively blind. However, these same latencies do not mask the discrete 
structural metadata logged for orchestration; instead, the leakage is driven by the \textbf{topological transpilation penalty}: the unavoidable inflation added by a compiler to make a program fit a restricted 
device coupling graph.

In this paper, we formalise this attack surface across six 
classification objectives. Our contributions are as follows:

\begin{enumerate}
\item \textbf{Identification and Formalisation:} We identify the 
post-cut metadata side-channel and formalise the structural attack surface 
exposed across standard cloud ingestion interfaces.

\item \textbf{Recovering Algorithmic Identity and Structure:} We define and 
evaluate six classification objectives, demonstrating that a provider can recover an 
algorithm's identity and its mathematical structure (\emph{Hamiltonian 
$k$-locality}) with near-perfect accuracy ($\text{AUC} = 1.000$) using 
only three structural metadata features.

\item \textbf{Empirical Dataset:} We provide a large-scale dataset of 
12{,}000 circuit fragments across eight algorithm families, transpiled 
against three hardware topologies.

\item \textbf{\ac{QPU} Validation:} We validate the attack on a 
156-qubit production system, confirming that true execution time is 
invariant across depth variations and isolating the transpilation penalty as 
the primary leakage mechanism.
\end{enumerate}

Metadata leakage must therefore be treated as a primary security
concern in cloud-based quantum systems, on the same footing as
protecting the quantum payload itself. The remainder of this paper
is organized as follows: Section~\ref{sec:related} reviews related
work; Sections~\ref{sec:system-model}--\ref{sec:threat-model}
establish the system and threat models respectively; Section~\ref{sec:mechanism} details the physical leakage mechanism; Sections~\ref{sec:attack_design}--\ref{sec:dataset} describe the attack methodology and empirical dataset; Sections~\ref{sec:evaluation}--\ref{sec:hardware-validation} present the evaluation results and hardware validation; and Sections~\ref{sec:results}--\ref{sec:conclusion} provide discussion on mitigations and concluding remarks.

\section{Related Work and Positioning}
\label{sec:related}

\textbf{Circuit Cutting and Knitting.}
Quantum circuit cutting was formalised by Peng~\emph{et
al.}~\cite{Peng_2020}, who introduced Pauli-based decompositions
enabling the execution of circuits exceeding device size limits.
CutQC~\cite{Tang_2021} provided an end-to-end automated partitioning
and reconstruction framework. Subsequent work refined estimator
efficiency through randomized measurements~\cite{Lowe_2023},
classical-communication-assisted knitting~\cite{Piveteau_2024}, and
generalised cutting interfaces~\cite{Schmitt_2025}. Across these
foundations, the emphasis is on reconstruction accuracy, variance
reduction, and hardware scalability. The confidentiality implications
of fragment-level orchestration are not their primary objective, leaving a critical gap regarding side-channel observability which is the focus of the present work.

\textbf{Circuit Cutting for Obfuscation and Security.}
Typaldos \emph{et al.}~\cite{10821443, 11250202} recently demonstrated that 
circuit cutting can be leveraged for algorithm obfuscation and outlined baseline 
security methodologies for cutting pipelines, while Bernardi \emph{et al.}~\cite{bernardi2025evaluatingsecuritypropertiesexecution} 
examined resilience properties under heuristic allocation. However, these 
frameworks explicitly leave open the question of whether the resulting 
compiled subcircuit structures could be exploited to reconstruct the original 
computational intent. This paper provides the definitive answer to that question: 
yes, they can, and with near-perfect accuracy using only three structural 
numbers logged for routine billing purposes.

\textbf{Quantum Confidentiality and Delegated Computation.}
Protocols such as \ac{UBQC}~\cite{broadbent2009ubqc},
QHE~\cite{Dulek_2018}, verifiable BQC~\cite{Fitzsimons_2017}, and
classical
verification~\cite{mahadev2023classicalverificationquantumcomputations}
were designed to hide the quantum \textit{data itself} from an
untrusted server. They provide strong guarantees for monolithic
execution or interactive protocols, but they do not address the
administrative metadata generated when a program is split into
fragments, the fragment count, the size of each piece, and
the order in which they are submitted. This paper fills those gaps.

\textbf{Quantum Cloud Side Channels.}
Researchers have shown that information can leak through several
physical channels: execution timing~\cite{lu2024quantumleaktimingsidechannel,
dong2025exploitingtimingsidechannelsquantum}, power drawn by the
controller~\cite{Erata_2024,xu2023explorationquantumcomputerpower},
scheduling patterns~\cite{9951250}, and crosstalk between users
sharing the same
QPU~\cite{choudhury2024crosstalkinducedchannelthreatsmultitenant,%
10.1145/3370748.3406570,lee2025swapattackstealthysidechannel}.
All of these attacks share a common assumption: the provider is
running one large, continuous program.

Circuit cutting breaks that assumption. When a program is split into
fragments, the provider no longer sees a single job, it sees a
structured collection of smaller jobs, each with its own width,
depth, and gate count, submitted in a specific order. The leakage
is no longer \textit{inside} the hardware; it is in the
\textit{paperwork} surrounding it. This paper is the first to
formally define and evaluate this new attack surface:
\gls{postcutinference}. \emph{We show that just three numbers, compiled width, depth, and 
two-qubit gate count, carry enough structural information to 
reconstruct the original computational intent with near-perfect 
accuracy.}

Notably, Dong~\emph{et
al.}~\cite{dong2025exploitingtimingsidechannelsquantum} demonstrate
ML-based timing extraction in simulation environments; our empirical
validation on production hardware shows this channel collapses to
chance, motivating the structural channel identified in this work.

\textbf{Relation to Classical Distributed Metadata Leakage.}
The broader pattern we exploit, workload partitioning generating 
classical metadata that leaks structural information, is not unique 
to quantum computing: distributed and multi-tenant classical systems 
leak similarly through scheduling records and network traffic shapes. 
What is specific here is the mechanism. The topological transpilation 
penalty arises from mapping a circuit's logical interaction graph onto 
a fixed, sparse qubit coupling map, a distortion with no direct 
classical analogue at this granularity, since classical schedulers do 
not reshape a program's gate structure to fit a processor's 
interconnect. Circuit cutting therefore inherits a familiar class of 
vulnerability but expresses it through a channel unique to constrained 
quantum hardware.

As summarized in Table~\ref{tab:comparison}, this work is the first to systematically evaluate post-cut transcript leakage under a semi-honest cloud model. To contextualize this threat, we contrast our approach against existing quantum confidentiality literature across five key evaluation dimensions: whether the framework explicitly addresses circuit \emph{cutting}, whether it provides cryptographic \emph{state protection} for the quantum payload, whether it analyzes or prevents classical \emph{metadata} leakage, whether the framework models an explicit \emph{attack vector} (offensive research), and whether the resulting orchestration is \emph{non-interactive}. Unlike standard cryptographic approaches which protect the quantum state inside the execution horizon, our analysis uniquely isolates the classical orchestration transcript.

\begin{table}[t]
\centering
\caption{\textbf{Comparison of quantum confidentiality literature.} Unlike standard cryptographic approaches (UBQC, QHE) which protect the quantum state inside the execution horizon, our analysis isolates the classical orchestration transcript exposed during cloud ingestion.}
\label{tab:comparison}
\small
\renewcommand{\arraystretch}{1.2}
\resizebox{\columnwidth}{!}{%
\begin{tabular}{lccccc}
\toprule
\textbf{Framework} &
\shortstack{\textbf{Cutting}\\\textbf{Target}} &
\shortstack{\textbf{State}\\\textbf{Protection}} &
\shortstack{\textbf{Metadata}\\\textbf{Protected}} &
\shortstack{\textbf{Attack}\\\textbf{Vector}} &
\shortstack{\textbf{Non-}\\\textbf{Interactive}} \\
\midrule
UBQC~\cite{broadbent2009ubqc}           & \none & \full & \none & \none & \none \\
QHE~\cite{Dulek_2018}                   & \none & \full & \none & \none & \full \\
Typaldos~\cite{10821443}                & \pfull & \none & \pfull & \none & \full \\
Bell \& Tr{\"u}gler~\cite{9951250}      & \none & \none & \pfull & \full & \full \\
Lu (Timing)~\cite{lu2024quantumleaktimingsidechannel} & \none & \none & \pfull & \full & \full \\
\midrule
\textbf{This Work}                      & \full & \full & \full & \full & \full \\
\bottomrule
\multicolumn{6}{l}{\footnotesize
  \full~Full \quad \pfull~Partial \quad \none~None}
\end{tabular}%
}
\end{table}

\section{System Model}
\label{sec:system-model}

We model circuit-cutting execution as a two-domain pipeline comprising
a \gls{trustedclient} and an \gls{untrustedcloud}. As detailed in the 
system architecture diagram of Fig.~\ref{fig:system_model}, the client owns 
circuit construction, partitioning, and post-processing, while the cloud 
owns job ingestion, scheduling, and execution. The two domains communicate 
only through the Submission Interface, the sole point at which information
crosses the trust boundary and the only surface visible to an adversary. 
Crucially, the client never hands the cloud a complete picture of the 
computation: the stitching specification, the reconstruction map, and the 
original unpartitioned circuit all remain local. What crosses the boundary 
is only the compiled fragment metadata needed to execute each job. Information 
crossing this boundary is partitioned into the \gls{submissiontranscript}~$\mathcal{L}$
and the \gls{executionartefact}~$\Pi$, both formally defined in
\S\ref{sec:threat-model}.

\paragraph{Execution Pipeline.}
The workflow proceeds through ten stages:
(1)~workload specification, (2)~circuit partitioning and preparation,
(3)~transcript preparation, (4)~job submission across the trust
boundary, (5)~cloud ingest and scheduling (where execution
artefacts~$\Pi$ first become observable), (6)~backend execution,
(7)~result serialisation, (8)~result retrieval, (9)~error mitigation
and stitching, and (10)~application output. The stitching
specification and reconstruction map remain local to the client
throughout and never cross the trust boundary.

\begin{figure}[t]
  \centering
  \includegraphics[width=1.0\columnwidth]{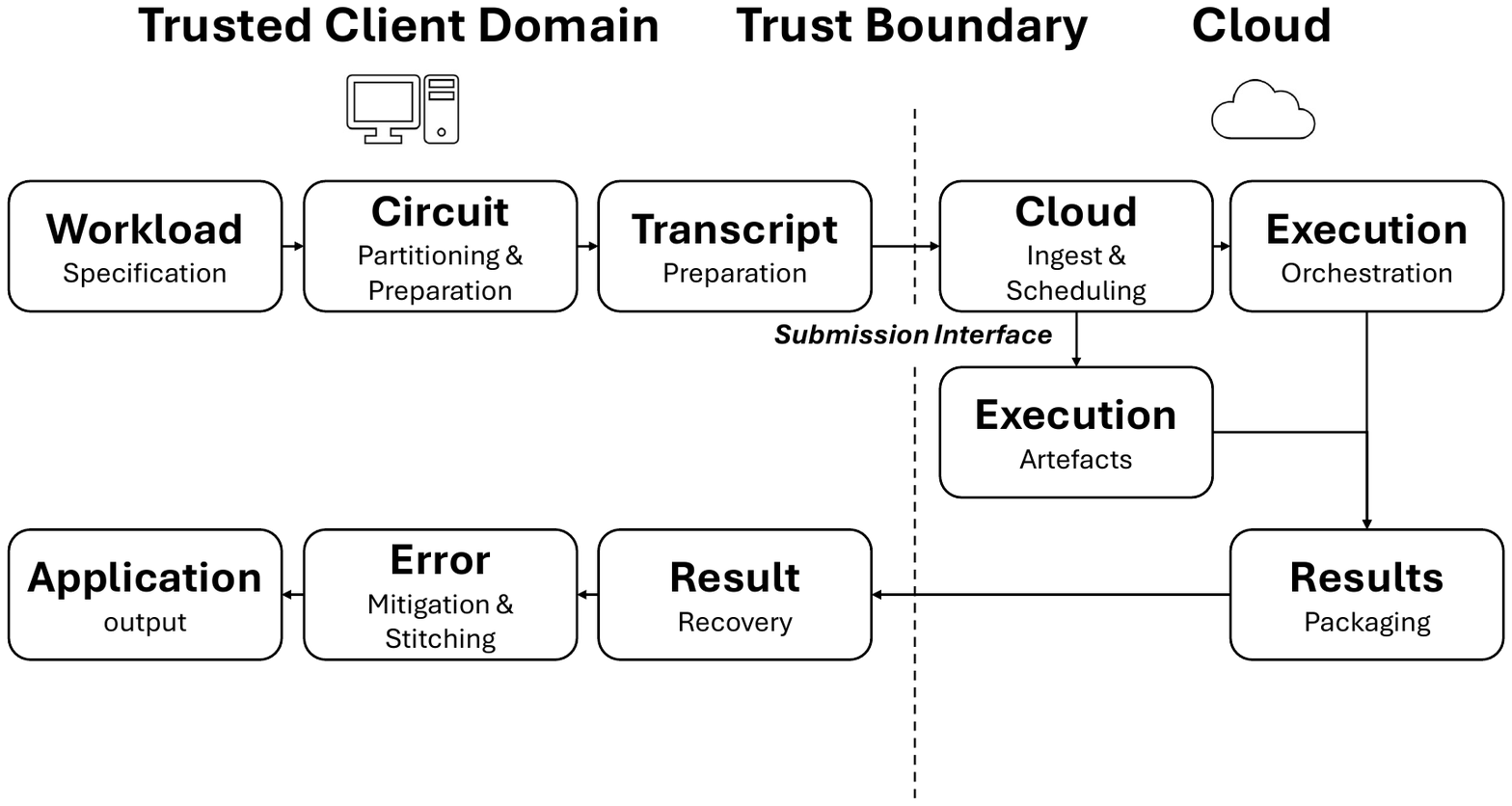}
  \caption{\textbf{System Architecture.} End-to-end circuit-cutting workflow bridging the trusted client and untrusted cloud domains.}
  \label{fig:system_model}
\end{figure}

\paragraph{Information Domains.}
We partition information into four domains:

\begin{itemize}[leftmargin=*]
  \item \textbf{Secret State~($\mathcal{S}$).} Exists solely within
    the trusted client domain: the original circuit~$G$,
    partitioning decisions, and the stitching specification.
    This is what the adversary seeks to recover.

  \item \textbf{Public Parameters~($\mathcal{P}$).} Known to both
    parties: hardware gate sets and device noise characteristics.

  \item \textbf{\Gls{observabletranscript}~($\mathcal{L}$).} The
    metadata surface exposed at the Submission Interface and during
    Cloud Ingest, comprising the prepared
    transcript~$\mathcal{L}$ and the execution artefacts~$\Pi$.
    This is the only information available to the adversary.

  \item \textbf{Ideal Assumptions.} We assume \textit{Execution
    Integrity} (faithful execution) and \textit{Client Isolation}
    (no leakage outside the defined interface), bounding the
    adversary to passive observation of $\mathcal{L}$ and $\Pi$.
\end{itemize}

\section{Threat Model}
\label{sec:threat-model}

\paragraph{Adversary Model.}
The provider acts as a \gls{semihonest} (honest-but-curious) adversary: it executes
fragment jobs correctly but passively logs the classical metadata
available at the Submission Interface. Its advantage derives solely
from offline statistical inference. It does \emph{not} observe
measurement outcomes, circuit payloads, compiled gate sequences,
qubit mapping, or physical-layer telemetry. 

\paragraph{Observable Transcript.}
The provider's observations are partitioned into two sub-transcripts
corresponding to distinct observability windows, an extraction process 
outlined in our threat model diagram, as shown in Fig.~\ref{fig:threat_model}. 
The submission-time transcript $\mathcal{L}$ is observable at the 
Submission Interface:
\[
\mathcal{L} = \{(w_i,\, d_i,\, q_i,\, s_i,\, t_i)\}_{i=1}^{N},
\]
where $w_i$, $d_i$, $q_i$ denote compiled width, depth, and 2Q gate 
count; $s_i$ the shot allocation; and $t_i$ the timestamp. The 
execution-time artefact $\Pi = \{(\pi_i,\, \tau_i)\}_{i=1}^{N}$ 
captures execution order and raw QPU duration. 

We establish the operational realism of these features in Table~\ref{tab:telemetry_mapping}, 
which maps our abstract transcript definitions directly to routine telemetry 
captured by production cloud services for billing and orchestration. For example, Amazon 
Braket logs per-task metadata including circuit structure and shot 
counts via Amazon CloudWatch and AWS 
CloudTrail~\cite{aws_braket_terms, aws_braket_cloudwatch}, while 
IBM Qiskit Runtime's billing model is explicitly parameterised by 
compiled circuit length and execution count, with granular per-job 
metrics retained by the service~\cite{ibm_runtime_workload, 
ibm_runtime_jobv2}. Consequently, the adversary requires no 
malicious modifications to the physical control stack or compilation 
pipeline to observe the transcript $\mathcal{L}$.

\begin{figure}[t]
  \centering
  \includegraphics[width=1.0\columnwidth]{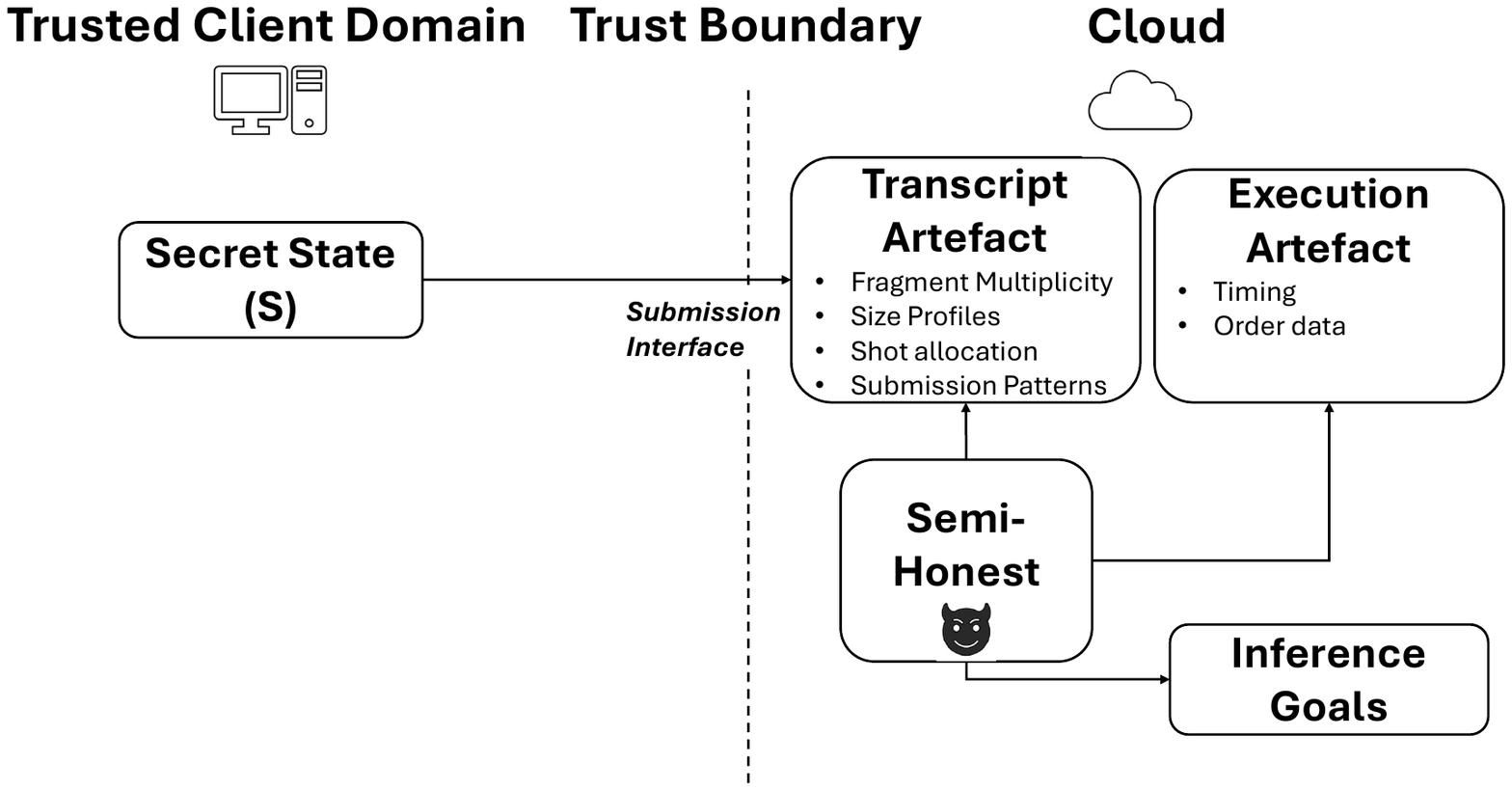}
  \caption{\textbf{Threat Model.} A semi-honest cloud provider passively extracts metadata from the observable transcript ($\mathcal{L} \cup \Pi$) to infer properties of the secret state~$\mathcal{S}$.}
  \label{fig:threat_model}
\end{figure}

\begin{table}[h]
\centering
\caption{Mapping of abstract transcript features to production cloud telemetry.}
\label{tab:telemetry_mapping}
\renewcommand{\arraystretch}{1.2}
\begin{tabular}{ll p{4.5cm}} 
\toprule
\textbf{Feature} & \textbf{Telemetry Source} & \textbf{Definition} \\
\midrule
$d_i$ & \texttt{tasks.depth} & Compiled circuit depth \\
$w_i$ & \texttt{tasks.width} & Active qubit count \\
$q_i$ & \texttt{tasks.2q\_count} & Compiled multi-qubit gates \\
$s_i$ & \texttt{tasks.shots} & User-defined shot allocation \\
$\tau_i$ & \texttt{execution\_spans} & Raw QPU runtime (nanoseconds) \\
$g_i$ & \texttt{job\_id / taskArn} & Job-group identifier (links the $K$ fragments of one cut) \\
\bottomrule
\end{tabular}
\end{table}

\paragraph{Adversary Assumptions and Limitations.}
Our threat model assumes a worst-case semi-honest provider under 
three specific operational conditions:
\begin{enumerate}[leftmargin=*]
    \item \textbf{Fragment Isolation:} The provider can group fragments 
    belonging to a single cutting job (of size $K$). This is 
    operationally feasible via shared \ac{API} tokens, correlation of 
    submission timestamps, or explicit job-grouping headers (e.g., 
    Braket's \texttt{taskArn} or Qiskit's \texttt{\detokenize{job_id}}).
    \item \textbf{Compiler Stability:} The provider's internal 
    transpilation behavior is deterministic and stable during the 
    attack window. We discuss the impact of compiler versioning 
    in \S\ref{sec:robustness}.
    \item \textbf{Passive Observation:} The adversary derives advantage 
    exclusively from offline statistical inference. Our results 
    thus represent an upper bound on leakage under perfect 
    segmentation.
\end{enumerate}

\subsection{Formal Security Model}

Intuitively, we want to ask: if a provider sees only the compiled 
metadata of a job, could it tell the difference between two 
candidate workloads a user might have run? We formalise this 
question as $\textsc{wsi-ind-}K$, an indistinguishability game between 
a challenger and the semi-honest provider acting as adversary $\mathcal{A}$.

The parameter $K$ is the fragment
multiplicity: the number of sub-circuits a workload is cut into and submitted
as one job. The game asks whether $\mathcal{A}$ can tell which of two candidate
workloads produced an observed transcript, using only the cloud-visible
metadata $\transcript = \observable \cup \artefacts$.

\begin{enumerate}[leftmargin=2em,label=\arabic*.]
  \item \textbf{Setup.} $\mathcal{A}$ submits two workload specifications
    $G_0, G_1$ of equal logical width, drawn from the public family
    set $\mathcal{F}$.
  \item \textbf{Challenge.} The challenger samples
    $b \xleftarrow{\$} \{0,1\}$, cuts $G_b$ into $K$ fragments under the
    deterministic client cutting policy, and transpiles each fragment
    against a backend in the public set $\mathcal{P}$.
  \item \textbf{Observation.} The challenger returns the transcript
    $\transcript_b = \{(w_i, d_i, q_i, s_i, t_i, \pi_i, \tau_i)\}_{i=1}^{K}$.
    No quantum state, measurement outcome, gate sequence, or qubit mapping
    is revealed.
  \item \textbf{Guess.} $\mathcal{A}$ outputs $b' \in \{0,1\}$ and wins
    if $b' = b$.
\end{enumerate}

The adversary's advantage is
$\mathsf{Adv}^{\mathrm{WSI}}_{\mathcal{A}}(K) = |\Prob[b' = b] - \tfrac{1}{2}|$.
A pipeline offers \emph{transcript indistinguishability} if this advantage is
negligible for every efficient $\mathcal{A}$; our results show the opposite.
The one-vs-rest Macro-AUC of \S\ref{sec:evaluation} is the multi-class
generalisation of this two-workload advantage: an $\mathrm{AUC}$ of $1.000$
corresponds to a distinguisher that wins the binarised game with certainty,
i.e.\ $\mathsf{Adv} = \tfrac{1}{2}$.

\paragraph{Why Metadata Leaks More Than Timing.}
The advantage above is carried by fragment \emph{shape}, not run \emph{time}.
Proposition~\ref{prop:leakage-decomp} makes this precise: timing is a
downstream function of the compiled fragment, so it cannot reveal more about
the secret than the structural metadata it derives from.

\begin{proposition}[Leakage Decomposition]
\label{prop:leakage-decomp}
Let \(\secret\) be the secret workload label, and let \(C\) be the full
compiled fragment produced by the transpiler. We write
\(\transcript = (\observable, \artefacts)\), where \(\observable\) is the
compiled structural metadata (width, depth, 2Q count) and \(\artefacts\) the
execution-time artefact (order, raw QPU duration). Since QPU duration is
generated by the control plane as a function of the compiled fragment, the
variables form the Markov chain
\(\secret \to C \to \artefacts\), whence
\[
  I(\secret; \artefacts) \;\leq\; I(\secret; C).
\]
In our implementation, the structural metadata \(\observable\) captures the
scheduling-relevant features of \(C\); the estimates below therefore bound the
leakage through the observable channel as
\(I(\secret; \artefacts) \le I(\secret; \observable)\).
\end{proposition}

\begin{proof}[Proof sketch]
Conditioned on the full compiled fragment \(C\), the runtime component of
\(\artefacts\) is independent of \(\secret\): the scheduler and pulse
generator observe only the compiled circuit, not the logical intent. This
gives the Markov chain, and the Data Processing
Inequality~\cite{cover2006elements} yields the bound. In practice, our
compiled metadata \(\observable\) accounts for the bulk of the
scheduling-dependent variation; the estimates of \S\ref{sec:channels}
instantiate this: \(I(\secret; \observable) \approx 2.89\) bits against a
\(3.0\)-bit maximum, while \(I(\secret; \artefacts) < 0.05\) bits, consistent
with the QPU timing invariance of \S\ref{sec:hardware-validation}.
\end{proof}

The adversary's advantage therefore comes from the structural metadata
captured by \(\observable\), and specifically from the overhead added when
fragments are adapted to the hardware's physical layout.

\paragraph{Adversary Capabilities.}
We assume a passive adversary who does not interfere with the computation. To ``learn" what different algorithms look like, the provider builds a reference library using two methods:
\begin{itemize}[leftmargin=*]
    \item \textbf{Known Benchmarks:} Running open-source quantum circuits to see what their metadata "fingerprints" look like.
    \item \textbf{Historical Logs:} Using previously labeled jobs from research papers or public repositories to train their inference models.
\end{itemize}
This allows the provider to recognize unseen user workloads by comparing them against their pre-trained library.

\begin{figure*}[!t]
  \centering
  \begin{subfigure}{\textwidth}
    \centering
    \includegraphics[width=0.95\textwidth]{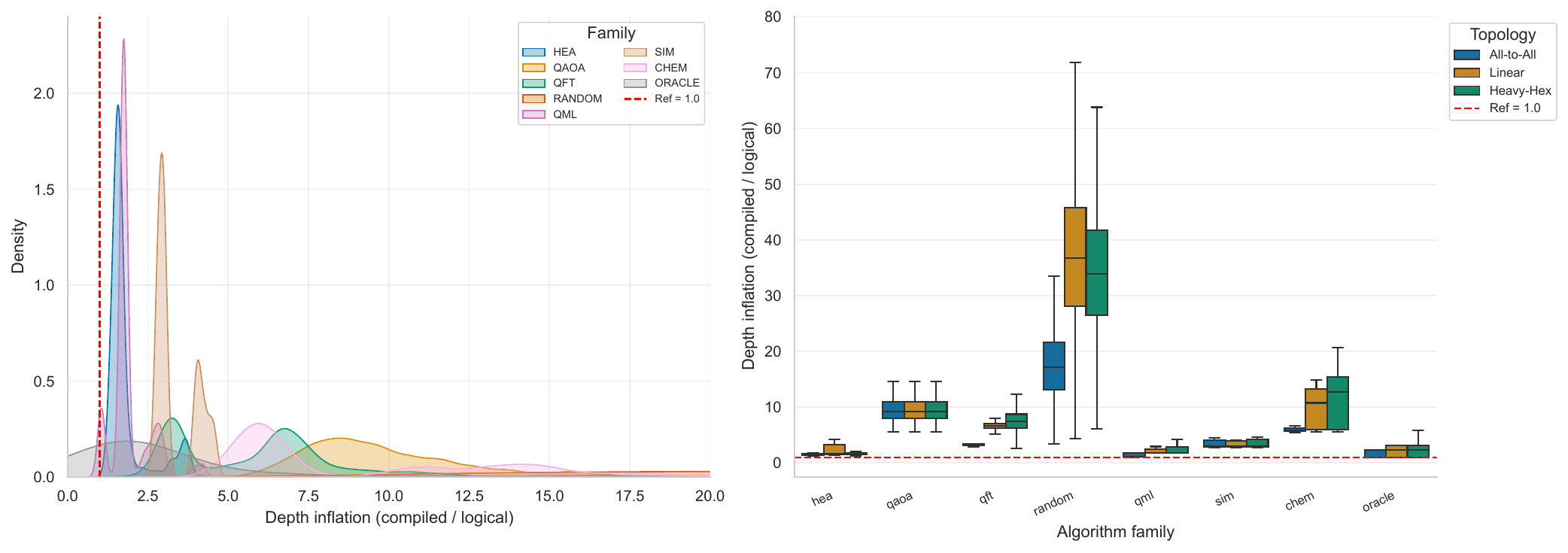}
    \caption{\textbf{Depth Inflation Characteristics.} Distribution of compilation-induced depth overhead across hardware topologies. The significant inflation in \textit{Heavy-Hex} and \textit{Linear} backends relative to the \textit{All-to-All} baseline (dashed red line) constitutes the primary structural leakage signal.}
    \label{subfig:depth_inflation}
  \end{subfigure}
  
  \vspace{1.2em} 
  
  \begin{subfigure}{\textwidth}
    \centering
    \includegraphics[width=0.95\textwidth]{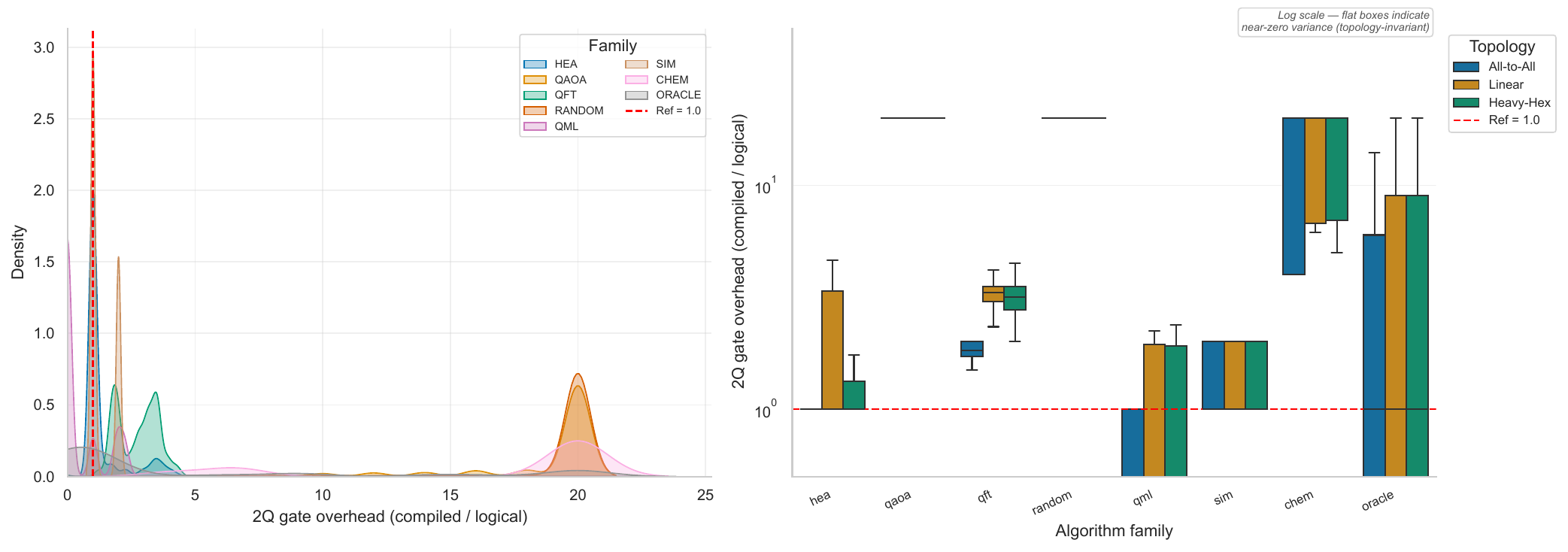}
    \caption{\textbf{Two-Qubit Routing Overhead.} Comparative density of additional 2Q gates required to satisfy connectivity constraints. The clear separation between algorithm families (e.g., QFT vs. HEA) underscores the fidelity of the transpilation penalty as an algorithmic fingerprint.}
    \label{subfig:routing_overhead}
  \end{subfigure}
  
  \caption{\textbf{Topological transpilation penalty analysis.} Distributions of (a)~depth inflation and (b)~two-qubit gate overhead across the three hardware topologies for the 12{,}000-fragment dataset.}
  \label{fig:routing_tax}
\end{figure*}

\begin{figure*}[t]
  \centering
  \includegraphics[width=\linewidth]{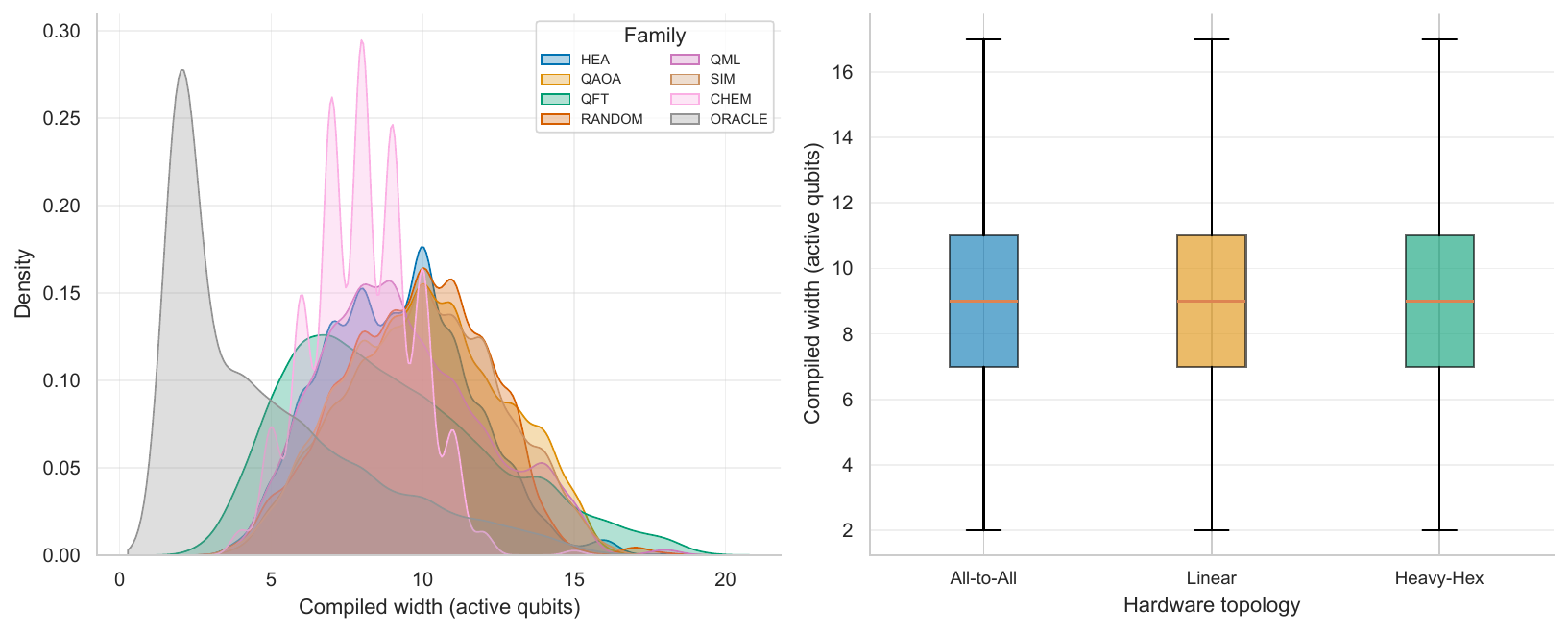}
\caption{Compiled width distributions: (a)~by algorithm family
(all backends pooled) and (b)~by hardware topology (all families
pooled). Width is family-discriminative but topology-invariant,
contrasting directly with depth inflation, as shown in
Fig.~\ref{fig:routing_tax}\subref{subfig:depth_inflation}.}
  \label{fig:width_distributions}
\end{figure*}

\section{Physical Mechanisms: Topological Transpilation Penalty}
\label{sec:mechanism}

The leakage of algorithmic identity through compiled metadata is driven by the \textbf{topological transpilation penalty}: the deterministic expansion of a circuit when its logical interaction graph is mapped onto a constrained physical coupling map. Because different algorithm families possess distinct logical connectivity requirements, the transpiler must insert varying quantities of SWAP logic to satisfy hardware adjacency constraints. We characterize this leakage through three complementary observables: depth inflation, 2Q gate inflation, and compiled width. Fig.~\ref{fig:routing_tax} shows how the first two separate the algorithm families across the three hardware topologies: hardware-restricted compilation deterministically encodes structural properties into the observable metadata, and the pronounced inflation on the Heavy-Hex and Linear backends forms the primary leakage signal.

\subsection{Characterizing Penalties by Algorithm Taxonomy}

As detailed in Table~\ref{tab:full_routing_tax}, the transpilation penalty is non-uniform across the 8-family taxonomy. We identify three distinct scaling profiles that serve as the primary discriminative signatures:

\begin{itemize}[leftmargin=*]
    \item \textbf{Low-Penalty Profiles (HEA, QML, Oracle):} These families exhibit the highest hardware compatibility. HEA and QML achieve this through hardware-aligned entanglement layers, while Oracle circuits (e.g., Grover, BV) maintain low transpilation penalties due to their sparse interaction graphs and minimal logical widths. On all-to-all architectures, their transpilation penalty is negligible ($<$10 extra 2Q gates), scaling only modestly as hardware connectivity becomes more restricted.
    
    \item \textbf{Topology-Invariant Profiles (QAOA, Sim):} These families represent Trotterized evolutions where the interaction structure remains fixed regardless of the target backend. For instance, QAOA maintains a constant mean penalty of 35.7 extra 2Q gates across all topologies, while Sim (Ising/Heisenberg) exhibits a perfectly stable penalty of 6.1 gates. This invariance allows the adversary to identify these workloads even when the underlying hardware architecture is obscured.

    \item \textbf{Topology-Sensitive Profiles (QFT, Chem, Random):} These families require high-degree or global connectivity. The QFT penalty nearly triples on heavy-hex compared to all-to-all, while Random circuits incur a massive penalty of 699.3 extra gates on linear backends. The explosive scaling of these metrics distinguishes global-entanglement algorithms from those with localized or linear interaction graphs.
\end{itemize}

\begin{table}[t]
\centering
\caption{Topological Transpilation Penalty Across the 8-Family Taxonomy. Values represent mean gate inflation across all sub-variants and fragment instances.}
\label{tab:full_routing_tax}
\small
\setlength{\tabcolsep}{8pt}
\renewcommand{\arraystretch}{1.1}
\begin{tabular}{clrr}
\toprule
\textbf{Topology} & \textbf{\makecell{Algorithm\\Family}} & \textbf{\makecell{Extra 2Q\\Gates}} & \textbf{\makecell{Depth\\Inflation}} \\
\midrule
\multirow{8}{*}[-1ex]{\rotatebox[origin=c]{90}{All-to-All}} 
 & HEA    & 0.0   & 1.5 \\
 & QML    & 0.0   & 1.5 \\
 & Sim    & 6.1   & 3.4 \\
 & Oracle & 8.7   & 3.7 \\
 & QAOA   & 35.7  & 9.8 \\
 & QFT    & 36.8  & 3.2 \\
 & Chem   & 148.3 & 6.0 \\
 & Random & 251.2 & 17.5 \\
\midrule
\multirow{8}{*}[-1ex]{\rotatebox[origin=c]{90}{Heavy-Hex}} 
 & HEA    & 5.0   & 1.7 \\
 & QML    & 62.8  & 2.2 \\
 & Sim    & 6.1   & 3.4 \\
 & Oracle & 14.4  & 4.8 \\
 & QAOA   & 35.7  & 9.8 \\
 & QFT    & 107.9 & 7.6 \\
 & Chem   & 286.0 & 11.8 \\
 & Random & 564.2 & 34.3 \\
\midrule
\multirow{8}{*}[-1ex]{\rotatebox[origin=c]{90}{Linear}} 
 & HEA    & 19.9  & 2.2 \\
 & QML    & 53.6  & 2.0 \\
 & Sim    & 6.1   & 3.3 \\
 & Oracle & 16.9  & 5.4 \\
 & QAOA   & 35.7  & 9.8 \\
 & QFT    & 103.8 & 6.5 \\
 & Chem   & 289.2 & 10.0 \\
 & Random & 699.3 & 37.5 \\
\bottomrule
\end{tabular}
\end{table}

\subsection{Compiled Width: An Invariant Structural Metric}

While the transpilation penalty is a variable metric dictated by hardware, \textbf{compiled width} (active qubit count) is an invariant property determined by the user's circuit-cutting partitioning strategy. 

As demonstrated in Section~\ref{sec:evaluation}, width serves as a primary discriminator. Because fragment width remains stable across different topologies—for example, Oracle fragments consistently cluster at 2--3 qubits—it allows the adversary to immediately isolate low-width logic from large-scale variational workloads. This invariance ensures that even if the hardware topology is optimized to minimize this transpilation penalty, the fundamental scale of the fragments still facilitates identity inference.

\subsection{Dimensionality of the Feature Set}

The three observables provide a multi-dimensional characterization of each fragment:
\begin{enumerate}
    \item \textbf{Depth Inflation} quantifies the mismatch between the logical graph and the physical topology.
    \item \textbf{2Q Gate Count} reflects the interaction density of the underlying Hamiltonian.
    \item \textbf{Compiled Width} exposes the logical scale and partitioning granularity.
\end{enumerate}
Collectively, these dimensions allow the inference model to resolve identities even when individual metrics between algorithm families overlap.

\section{Attack Design and Methodology}
\label{sec:attack_design}

\subsection{Attacker Models}

We restrict our attacker models to tree-based ensembles. The
transcript features are low-dimensional (five raw fields per
fragment), heterogeneously scaled (width and depth span different
numeric ranges than shot counts and timestamps), and tabular rather
than sequential or image-like. Tree ensembles handle this profile
natively, without requiring feature normalisation or a large
training set to find effective split points, and their inductive
bias favours axis-aligned thresholds that map naturally onto
structural features like depth and gate count. We therefore evaluate
three ensemble variants to check that the leakage signal is not an
artefact of a single model's inductive bias, rather than presenting
one classifier as the maximal achievable attacker:

\begin{itemize}[leftmargin=*]
  \item \textbf{Random Forest (RF)}: headline attacker. Robust to
    feature scaling, strong on tabular data.
  \item \textbf{ExtraTrees (ET)}: stronger non-linear baseline, using
    randomised split thresholds to reduce variance.
  \item \textbf{HistGradientBoosting (HGB)}: gradient boosting
    baseline, included to confirm the signal is not specific to
    bagging-based ensembles.
\end{itemize}

As shown in Table~\ref{tab:attacker_comp}, all three converge to
within 0.005 AUC of each other on every task, confirming the
leakage is a property of the transcript rather than of any one
model family.

All training is offline and computationally feasible (under 8 minutes
on commodity hardware), consistent with a provider who caches
metadata and trains models asynchronously.

\subsection{Evaluation Protocols}

\begin{itemize}[leftmargin=*]
  \item \textbf{\Gls{instancedisjoint}:} No cutting job appears in
    both train and test sets. Prevents memorisation of specific
    circuit instances.

  \item \textbf{\Gls{sizeholdout}:} The largest-width quartile is
    held out as test. Tests generalisation to unseen workload scales.

  \item \textbf{\Gls{matchedfootprint}:} A caliper of 0.20 in
    normalised (width, depth, fragment count) space restricts
    evaluation to footprint-matched subsets. Removes coarse scale
    cues and tests whether residual leakage survives equalisation.
\end{itemize}

We report Accuracy~(ACC), Macro-F1, and Macro-AUC~(\ac{OvR}) with 95\%
bootstrap confidence intervals ($B{=}1000$). These specific metrics were selected to handle the multi-class nature of our inference objectives: Macro-AUC ensures that classification performance is evaluated equally across all structural sub-variants, preventing dominant algorithm families from artificially inflating the success rate. For binary W1 cut mechanism, AUC uses the positive-class probability directly.

\section{Empirical Dataset Construction}
\label{sec:dataset}

\subsection{Generative Pipeline}
The dataset is constructed via a standardized three-stage pipeline: (i)~\textbf{Logical Specification}, where the core algorithm class is defined; (ii)~\textbf{Sub-variant Assignment}, where structural modifications (e.g., connectivity patterns) are applied to ensure diversity; and (iii)~\textbf{Topology-Restricted Transpilation}, where logical circuits are mapped to specific hardware graphs. 

We generate 12,000 unique logical fragments across eight algorithm families, expanded to 36,000 records by transpiling each fragment against three distinct hardware topologies. To maintain a realistic experimental context, fragments are grouped into cutting jobs of $K{=}6$. Specific distribution counts and sub-variant definitions are detailed in Table~\ref{tab:dataset_dist} (see Appendix~\ref{app:dataset}); family counts are balanced to ensure the inference models learn structural patterns rather than class-frequency biases.

\subsection{Structural Classification Regimes}
Rather than treating families in isolation, we categorize the dataset into three scaling regimes based on their entanglement density and resulting ``transpilation penalty.''

\paragraph{Variational and Shallow Ansätze} 
This regime includes \textbf{HEA, \ac{QML}, and Sim} (Trotterized simulation). These circuits are characterized by local, often linear connectivity. Because their logical structure closely aligns with physical hardware graphs, they incur minimal routing overhead. They represent the baseline of the metadata channel, where compiled features remain closest to logical specifications.

\paragraph{Globally Entangled Structures} 
This regime comprises \textbf{QFT, Chem, and Random} circuits. These algorithms rely on dense, long-range interactions that do not map naturally to restricted hardware topologies (like heavy-hex). Consequently, they trigger extensive SWAP-gate insertion during compilation. This massive, deterministic depth inflation serves as the primary discriminative signal for our attack.

\paragraph{Intermediate Structures} 
\textbf{QAOA and Oracle} circuits occupy a middle ground. While they possess more complex connectivity than variational ansätze, their structures often contain Hamiltonian-specific symmetries that allow for partial topological alignment. Their metadata signatures exhibit moderate depth scaling that sits between the shallow and global regimes.

\subsection{Label Assignment and Statistical Alignment}
The six inference labels, A1 (algorithm family), W1 (cut mechanism), 
W2 (backend topology), and H1--H3 (Hamiltonian connectivity, geometry, 
and $k$-locality), describe properties of the original \emph{logical} 
circuit and are assigned before compilation. The model must therefore 
recover these properties from compiled metadata alone, demonstrating 
that any leakage we observe reflects genuine algorithmic intent rather 
than compiler-specific artifacts.

To check that our simulated dataset is a fair stand-in for real 
workloads, we compare its distributions against production telemetry 
from \S\ref{sec:hardware-validation}. Table~\ref{tab:validation} reports 
this comparison using two complementary distance measures: the 
Kolmogorov--Smirnov statistic (KS), which captures the maximum gap 
between cumulative distributions, and the 1-Wasserstein distance 
($W_1$, also known as the Earth Mover's distance), which captures the 
average displacement needed to align one distribution with the other. 
Lower values indicate closer agreement for both.

Alignment is strongest for circuit width: KS=0.148, with means agreeing 
to within 0.06 qubits. Compiled depth shows moderate agreement 
(KS=0.197). 2Q gate count shows similar moderate agreement (KS=0.185). 
For both depth and 2Q, however, the dataset mean runs notably higher 
than the hardware mean.

This mean gap is expected and reflects how the two datasets were 
constructed. The dataset deliberately samples multiple structural 
sub-variants per family, for example, Random shallow/medium/deep, 
or HEA linear/circular/reverse, producing a long right tail of 
high-overhead circuits. The hardware experiment, by design, runs a 
single representative fragment per family at each width. Means are 
sensitive to this tail; medians are not. The median column shows 
that the bulk of the dataset sits in the same regime as production 
hardware on all three observables, confirming the dataset spans the 
operating point of real workloads rather than diverging from it.

We also performed a \textit{depth-restricted evaluation} to confirm 
signals remain robust within empirical hardware bounds. As shown in 
\S\ref{sec:robustness}, the attack maintains near-perfect accuracy 
(AUC 1.000) when restricted to the exact depth range observed on 
production hardware (18 to 400 gates). This confirms that the transpilation penalty is a robust physical signal, not a simulation artifact.

\begin{table}[t]
\centering
\caption{Distributional alignment between simulated dataset and 
\texttt{ibm\_marrakesh} telemetry. Lower $W_1$ and KS indicate 
closer agreement.}
\label{tab:validation}
\begin{tabular}{lrrrrrr}
\toprule
& \multicolumn{2}{c}{\textbf{Distance}} & 
\multicolumn{2}{c}{\textbf{Mean}} & 
\multicolumn{2}{c}{\textbf{Median}} \\
\cmidrule(lr){2-3} \cmidrule(lr){4-5} \cmidrule(lr){6-7}
\textbf{Metric} & $W_1$ & KS & dataset & Emp. & dataset & Emp. \\
\midrule
Width  & 0.87  & 0.148 & 8.61   & 8.55   & 9.0  & 8.5  \\
Depth  & 56.25 & 0.197 & 137.36 & 129.36 & 57.0 & 80.0 \\
2Q     & 63.45 & 0.185 & 126.59 & 66.33  & 40.0 & 28.5 \\
\bottomrule
\end{tabular}
\end{table}

\FloatBarrier

\section{Inference Evaluation}
\label{sec:evaluation}

\begin{table}
\centering
\caption{Headline results: instance-disjoint~(ID) and
size-holdout~(SH) performance (RF, full features, 95\%
bootstrap CIs).}
\label{tab:headline}
\footnotesize
\setlength{\tabcolsep}{3pt}
\renewcommand{\arraystretch}{1.1}
\begin{tabular}{@{}lcccccc@{}}
\toprule
\multirow{2}{*}{\textbf{Task}} &
\multicolumn{3}{c}{\textbf{ID}} &
\multicolumn{3}{c}{\textbf{SH}} \\
\cmidrule(lr){2-4}\cmidrule(l){5-7}
& ACC & F1 & AUC & ACC & F1 & AUC \\
\midrule
A1: Family (8-way)       & 0.992 & 0.991 & 1.000 & 0.900 & 0.818 & 0.987 \\
W1: Cut Mech.\ (2-way)   & 0.941 & 0.941 & 0.991 & 0.957 & 0.957 & 0.994 \\
W2: Backend (3-way)      & 0.609 & 0.609 & 0.818 & 0.510 & 0.507 & 0.739 \\
H1: Connect.\ (3-way)    & 0.941 & 0.938 & 0.994 & 0.955 & 0.950 & 0.993 \\
H2: Geometry (3-way)     & 0.955 & 0.950 & 0.996 & 0.940 & 0.886 & 0.983 \\
H3: $k$-Locality (3-way) & 0.995 & 0.995 & 1.000 & 0.847 & 0.833 & 0.986 \\
\bottomrule
\end{tabular}
\end{table}

To quantify the severity of the post-cut metadata side channel, we
systematically evaluate the adversary's ability to infer hidden
workload properties from the observable transcript. This section
details the performance of our models across the six target
objectives, breaking down the specific structural signatures driving
the leakage.

\subsection{Headline Leakage Under Instance-Disjoint Evaluation}
\label{sec:headline}

Table~\ref{tab:headline} reports the headline inference results.
Under instance-disjoint generalisation, post-cut transcript metadata
supports high-confidence inference across all six objectives, with
algorithm family~(A1) and $k$-locality~(H3) achieving perfect
separability (Macro-AUC $1.000$), and all remaining tasks exceeding
AUC~$0.99$ with the exception of backend topology inference.

\begin{figure*}[!t]
  \centering
  \begin{subfigure}{0.32\textwidth}
    \includegraphics[width=\textwidth]{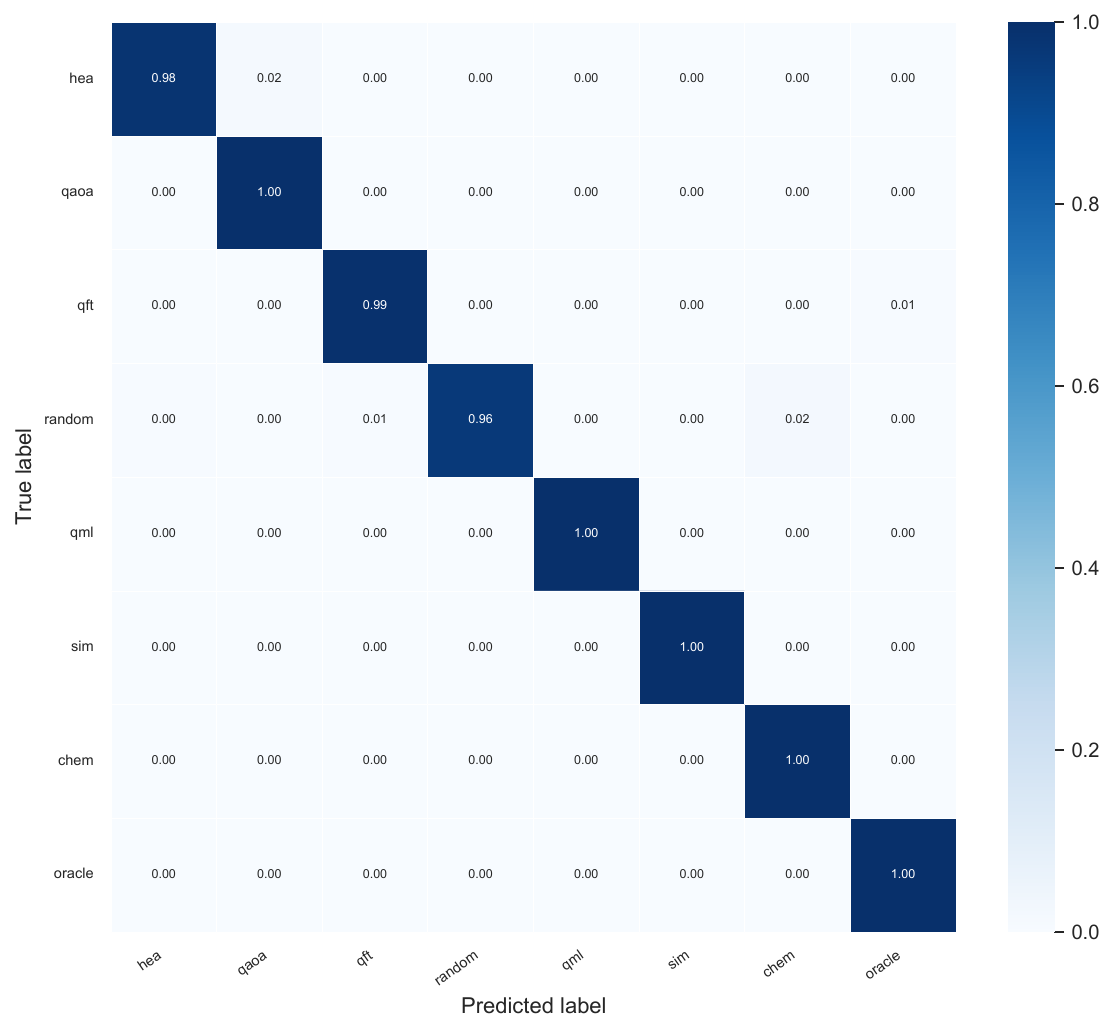}
    \caption{A1: Algorithm family}
    \label{subfig:a1_cm}
  \end{subfigure}\hfill
  \begin{subfigure}{0.32\textwidth}
    \includegraphics[width=\textwidth]{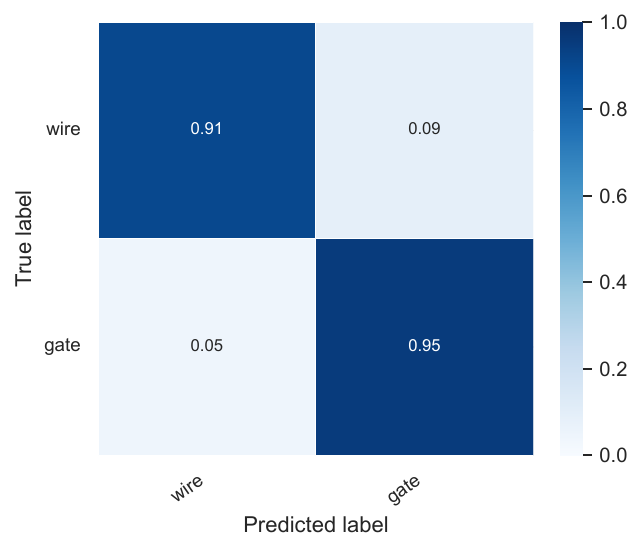}
    \caption{W1: Cut mechanism}
    \label{subfig:w1_cm}
  \end{subfigure}\hfill
  \begin{subfigure}{0.32\textwidth}
    \includegraphics[width=\textwidth]{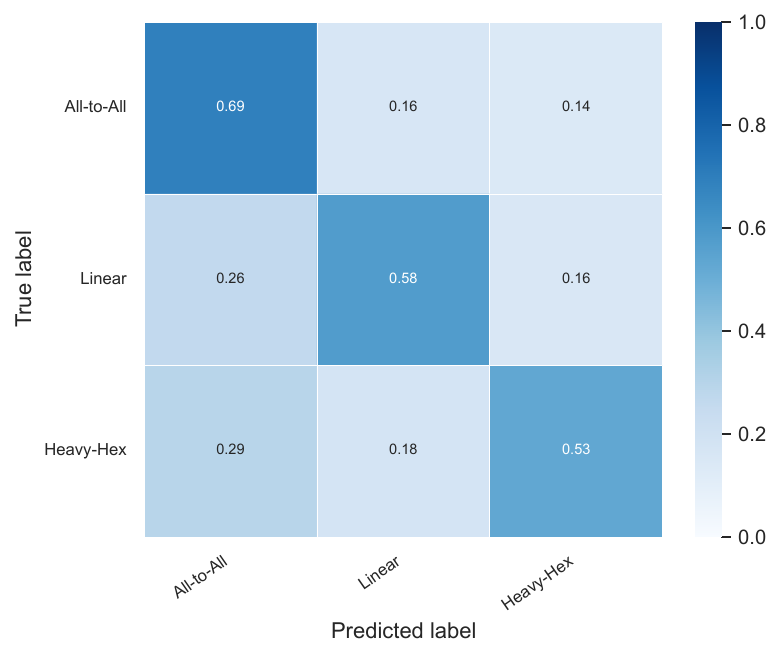}
    \caption{W2: Backend topology}
    \label{subfig:w2_cm}
  \end{subfigure}
  \vspace{1.5em}
  \begin{subfigure}{\textwidth}
    \includegraphics[width=\textwidth]{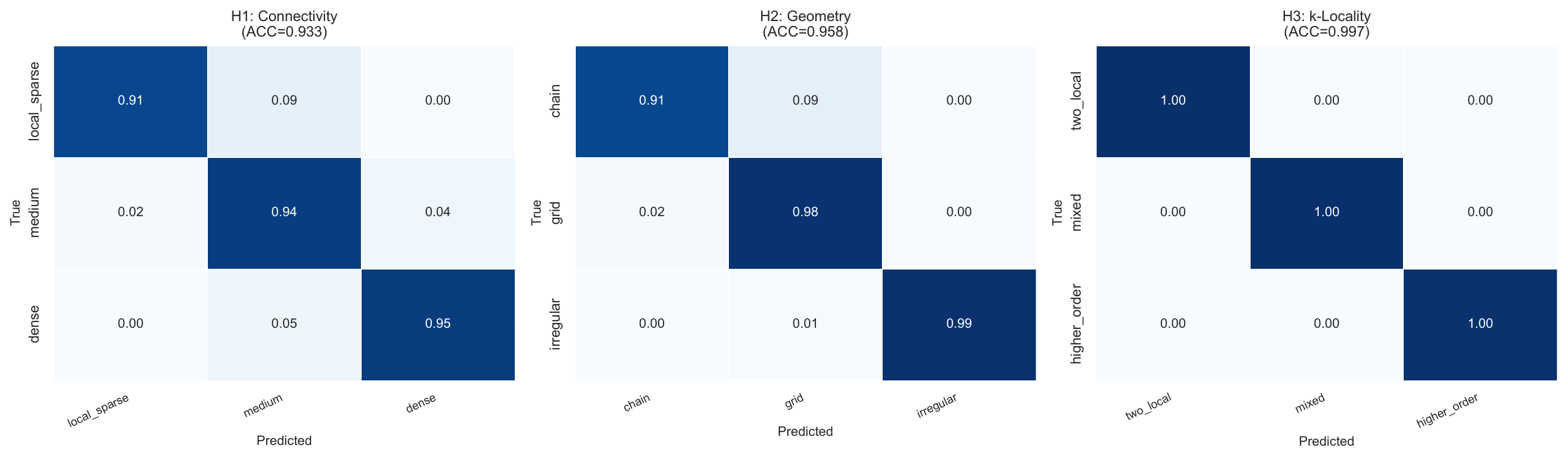}
    \caption{H1--H3: Hamiltonian structure}
    \label{subfig:h_cm}
  \end{subfigure}
  \caption{Inference evaluation confusion matrices across
  classification tasks (a--d).}
  \label{fig:evaluation_dashboard}
\end{figure*}

The size-holdout results reveal a strong generalisation profile across
all tasks. A1 size-holdout achieves AUC~$0.987$ (ACC~$0.900$),
confirming that algorithm family leakage generalises consistently to
unseen circuit scales. The Hamiltonian structure tasks (H1--H3) and
W1 are particularly stable, with size-holdout AUC remaining within
$0.013$ of instance-disjoint performance for H1 and W1, confirming
their leakage is scale-robust.

Based on these results, a clear performance hierarchy emerges. Under
instance-disjoint evaluation, the hierarchy is:
\[
\text{H3} \approx \text{A1} > \text{H2} \approx \text{H1} \approx
\text{W1}_{\text{cut}} \gg \text{W2}_{\text{backend}}
\]
Under size-holdout evaluation, the scale-robust cutting and
Hamiltonian tasks dominate:
\[
\text{W1}_{\text{cut}} \approx \text{H1} \approx \text{H2} >
\text{H3} \approx \text{A1} \gg \text{W2}_{\text{backend}}
\]
This ordering reflects how directly each hidden property maps onto
the structural compiled features. Tasks where the transpilation penalty
produces unambiguous, structure-only signatures achieve near-ceiling
performance. Backend topology inference (W2) sits at the bottom of
the hierarchy; because restricted hardware topologies impose
overlapping inflation penalties on all complex algorithms, isolating
the hardware itself from the algorithm running on it remains
significantly more difficult. Nonetheless, W2 now achieves
AUC~$0.818$ under instance-disjoint evaluation, well above the
$0.333$ random baseline, reflecting the larger and more structurally
diverse dataset.

\subsection{Confusion Analysis}
\label{sec:confusions}

\begin{table*}[t]
\centering
\caption{Feature channel ablation: Macro-AUC under instance-disjoint
evaluation (RF). Structure dominates for all tasks. Timing-only and
shots-only collapse to chance across all tasks (timing:
$0.479$--$0.504$; shots: $0.500$), confirming that leakage is
carried exclusively by compiled structural metadata. \texttt{structure\_only}
exceeds the full feature set on W2, the only task where additional
channels add noise rather than signal. Chance baselines: $1/2$ for
W1, $1/3$ for W2 and H1--H3, $1/8$ for A1.}
\label{tab:ablation}
\setlength{\tabcolsep}{5pt}
\renewcommand{\arraystretch}{1.1}
\begin{tabular}{lrrrrrr}
\toprule
\textbf{Task}
  & \textbf{Full}
  & \textbf{Structure only}
  & \textbf{Timing only}
  & \textbf{Shots only}
  & \textbf{No timing}
  & \textbf{No shots} \\
\midrule
A1: Algorithm family  & 1.000 & 1.000 & 0.504 & 0.500 & 1.000 & 1.000 \\
W1: Cut mechanism     & 0.991 & 0.991 & 0.487 & 0.500 & 0.990 & 0.990 \\
W2: Backend topology  & 0.818 & \textbf{0.829} & 0.497 & 0.500 & 0.826 & 0.816 \\
H1: Connectivity      & 0.994 & 0.994 & 0.479 & 0.500 & 0.994 & 0.994 \\
H2: Geometry          & 0.996 & 0.996 & 0.497 & 0.500 & 0.996 & 0.996 \\
H3: $k$-Locality      & 1.000 & 1.000 & 0.496 & 0.500 & 1.000 & 1.000 \\
\bottomrule
\end{tabular}
\end{table*}

To elucidate the specific decision boundaries learned by the models,
we analyse the instance-disjoint confusion matrices, presented in the
top and middle rows of Fig.~\ref{fig:evaluation_dashboard}.

As shown in the top row of Fig.~\ref{fig:evaluation_dashboard}, the
algorithm family~(A1) matrix shows perfect separability for Oracle
and Sim circuits~(1.00 recall), reflecting their highly distinctive
compiled footprints. In contrast, HEA and QAOA exhibit a small mutual
confusion rate, physically consistent given their shared reliance on
shallow, linear entanglement structures. For the W1 cut mechanism, classification errors are highly symmetric, 
indicating the classifier has learned a robust physical boundary rather 
than merely exploiting class imbalances. The W2 backend topology matrix confirms our
hierarchy findings: while All-to-All routing is well-isolated due to
its zero-overhead mapping, Linear and Heavy-Hex topologies remain
confused, consistent with both restricted topologies imposing
similarly severe depth and gate inflation penalties.

The middle row of Fig.~\ref{fig:evaluation_dashboard} details the
Hamiltonian structure objectives. Connectivity~(H1) exhibits classic
ordinal confusion: the intermediate \texttt{medium} class bleeds into
its neighbours, whereas the structural extremes,
\texttt{local\_sparse} and \texttt{dense}, are cleanly isolated.
Geometry~(H2) displays an asymmetric confusion, where grid topologies
frequently misclassify as chains, reflecting the physical artefact of
cutting that causes moderate-width grid circuits to compile into
narrow shapes indistinguishable from 1D chains. Finally,
$k$-locality~(H3) achieves near-perfect separation across all classes,
confirming that local interaction density leaves an exceptionally
strong, topology-invariant signature on the transcript.

\subsection{Channel Decomposition}
\label{sec:channels}

To identify which components of the execution transcript carry the
most information, we performed a feature ablation study across all
six inference tasks. Results are reported in Table~\ref{tab:ablation}.

The \texttt{structure\_only} mask confirms that compiled structural
metadata dominates leakage across all tasks, achieving AUC of 1.000
for A1 and H3, 0.991 for W1, and 0.994 for H1. This dominance is
physically coherent: the transpilation penalty mechanism operates through
compiled circuit shape rather than submission metadata.

\begin{figure*}[t]
  \centering
  \includegraphics[width=\linewidth]{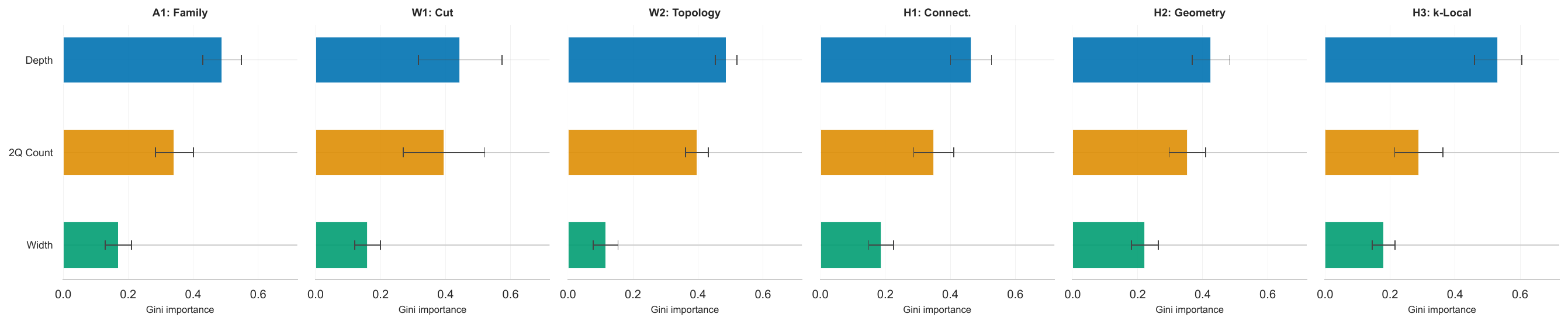}
  \caption{Random Forest Gini importance (mean~$\pm$~1\,SD across
  trees) for all six inference tasks. Compiled depth is the dominant
  feature across tasks, consistent with the transpilation penalty mechanism.
  Width contributes most for topology inference.}
  \label{fig:importance}
\end{figure*}

\begin{figure*}[t]
  \centering
  \includegraphics[width=\textwidth]{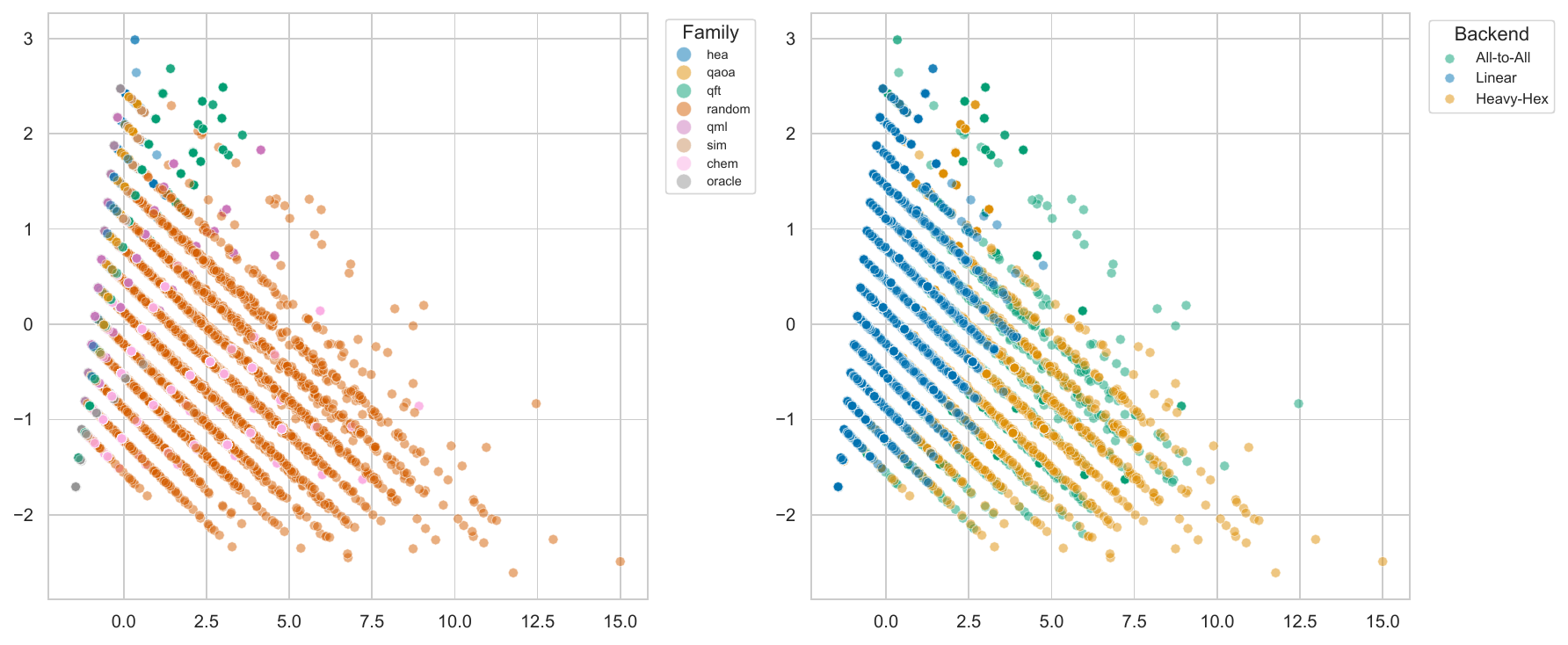}
  \caption{PCA of compiled features (width, depth, 2Q count). PC1
  captures 70.3\% of variance. Left: Sim and Oracle are isolated;
  Random diffuses along PC1. Right: All-to-All separates from
  restricted topologies; Linear and Heavy-Hex overlap, directly
  explaining the modest topology inference accuracy.}
  \label{fig:pca}
\end{figure*}

Both \texttt{timing\_only} and \texttt{shots\_only} channels collapse
to chance-level performance across all six tasks (timing:
$0.479$--$0.504$; shots: exactly $0.500$). This confirms two
independent null results: pure submission timing carries no
algorithmic signal, consistent with the QPU execution invariance
demonstrated in \S\ref{sec:hardware-validation}; and shot-count
statistics carry no structural information, confirming that the
observed leakage is not an artefact of any shot allocation policy.
The structural leakage finding is therefore unambiguous: the transpilation penalty encodes all recoverable information, and no non-structural channel
contributes meaningfully to the attack.

The one exception is W2 backend topology inference, where
\texttt{structure\_only} (0.829) \emph{exceeds} the full feature set
(0.818), the only task where removing additional channels improves
performance. This indicates that timing and shot features introduce
noise for topology inference, and that a structural-only attacker is
strictly more effective on the hardest task.

\subsection{Classification Efficacy and Transcript Separability}

Rather than capturing the full theoretical entropy of the hidden variable, the structural transcript $L$ provides highly optimized, discriminative boundaries that enable near-perfect classification. The topological transpilation penalty yields an $\text{AUC} = 1.000$ across the 8-way algorithm family taxonomy (A1), with classifiers exploiting isolated structural correlations to achieve this separability. For instance, the compiled footprint cleanly separates the deep algebraic global entanglement of QFT circuits---which exhibit quadratic depth scaling---from the minimal, near-constant footprints of Oracle circuits (e.g., depth = 3 for Deutsch-Jozsa) and the chaotic, high-variance stress-testing of Random circuits. The only non-negligible confusion at the family level occurs between Random and Chemistry variants (a 2\% error rate), reflecting their overlapping high-overhead regimes. 

Consequently, the model achieves near-perfect predictive accuracy strictly through empirical feature separability, as validated by confusion matrices and feature importance analysis, without requiring near-maximum mutual information. However, while perfect classification is achieved at the family level, sub-family discrimination (A2) remains challenging for variational families. As illustrated in Fig.~\ref{fig:subfamilies}, HEA sub-variant accuracy drops to 0.66, indicating that while the top-level algorithmic structure is strongly imprinted, fine-grained structural variants are not always recoverable.

\subsection{Matched-Footprint Control}
\label{sec:footprint}

\begin{table}[h]
\centering
\caption{Natural vs. matched-footprint AUC (instance-disjoint, RF,
caliper~$=0.20$, 1206/6000 samples retained). $\ddagger$~AUC is 
mathematically undefined for A1 in this subset due to single-class 
test folds at this caliper level; however, leakage persistence is 
confirmed by H1--H3 and W2 results.}
\label{tab:matched}
\begin{tabular}{lrr}
\toprule
\textbf{Task} & \textbf{Natural AUC} & \textbf{Matched AUC} \\
\midrule
A1: Family        & 1.000 & N/A$^{\ddagger}$ \\
W1: Cut mechanism & 0.991 & 0.967 \\
W2: Backend       & 0.818 & 0.928 \\
H1: Connectivity  & 0.994 & 0.992 \\
H2: Geometry      & 0.996 & 0.987 \\
H3: $k$-Locality  & 1.000 & 0.998 \\
\bottomrule
\end{tabular}
\end{table}

A common vulnerability in metadata-based machine learning models is
the over-reliance on trivial scale artefacts (e.g., identifying a QFT
circuit simply because it is larger than an Oracle circuit). To ensure
our attacker is learning genuine structural signatures rather than
coarse size distributions, we adapt \emph{caliper matching}, a
technique from the causal inference literature originally used for
confound control in observational studies~\cite{rosenbaum1985caliper},
to our evaluation setting. Using a caliper of 0.20 in normalised
width, depth, and fragment count space, we restrict evaluation
exclusively to similarly sized sub-populations.

As detailed in Table~\ref{tab:matched}, performance is either
maintained or increases after footprint equalisation. The H3
($k$-locality) retains an AUC of 0.998, and H1, H2, and W1 all
decline by at most 0.024, confirming that these structural signatures
are robust and independent of scale. Notably, W2 backend topology
inference \emph{increases} from 0.818 to 0.928 after footprint
matching. This reversal indicates that at matched scales, compiled
structural features provide a cleaner discriminative signal for
topology than when the full size distribution introduces
within-class variance. It reinforces that topology leakage is a
genuine structural phenomenon, not an artefact of coarse scale
differences.

\subsection{Feature Importance}
\label{sec:importance}

Fig.~\ref{fig:importance} visualises the Gini feature importance
extracted from the Random Forest models. Compiled depth emerges as
the primary feature across all six tasks, carrying an importance
weighting of 0.44--0.51. This strongly validates our physical
premise: depth inflation is the most direct and measurable
manifestation of the transpilation penalty, acting as the primary mechanism
for distinguishing algorithm families.

Compiled active width gains relative importance specifically for
topology inference (W2). While routing overhead aggressively mutates
circuit depth, the number of active qubits required to execute a
fragment remains a relatively static, topology-invariant fingerprint
of the client's original partitioning strategy. Nonetheless, all
three structural features contribute meaningfully; removing any
single dimension measurably degrades the attacker's discriminative
power.

\subsection{Metadata Latent Space}
\label{sec:pca}

To visualize the separability of the dataset, Fig.~\ref{fig:pca}
projects the high-dimensional transcript metadata into a 2D space
using Principal Component Analysis (PCA). The first principal
component (PC1) captures 70.3\% of the total variance, functioning
broadly as an axis of aggregate circuit complexity.

The left panel illustrates how Sim and Oracle circuits form distinct,
tightly grouped islands, reflecting their highly consistent and
restricted compiled footprints. In contrast, Random circuits diffuse
entirely along the PC1 axis, occupying the high-complexity tail. 
We note that the distinct striped or banded structures visible in these projections 
are natural artefacts of the discretized, integer-valued features (such as depth, 
2Q gate count, and active width) underlying the structural metadata, rather than 
isolated continuous sub-manifolds.

\section{Hardware API Validation and Scaling Analysis}
\label{sec:hardware-validation}

\begin{figure*}[t]
  \centering
  \includegraphics[width=\linewidth]{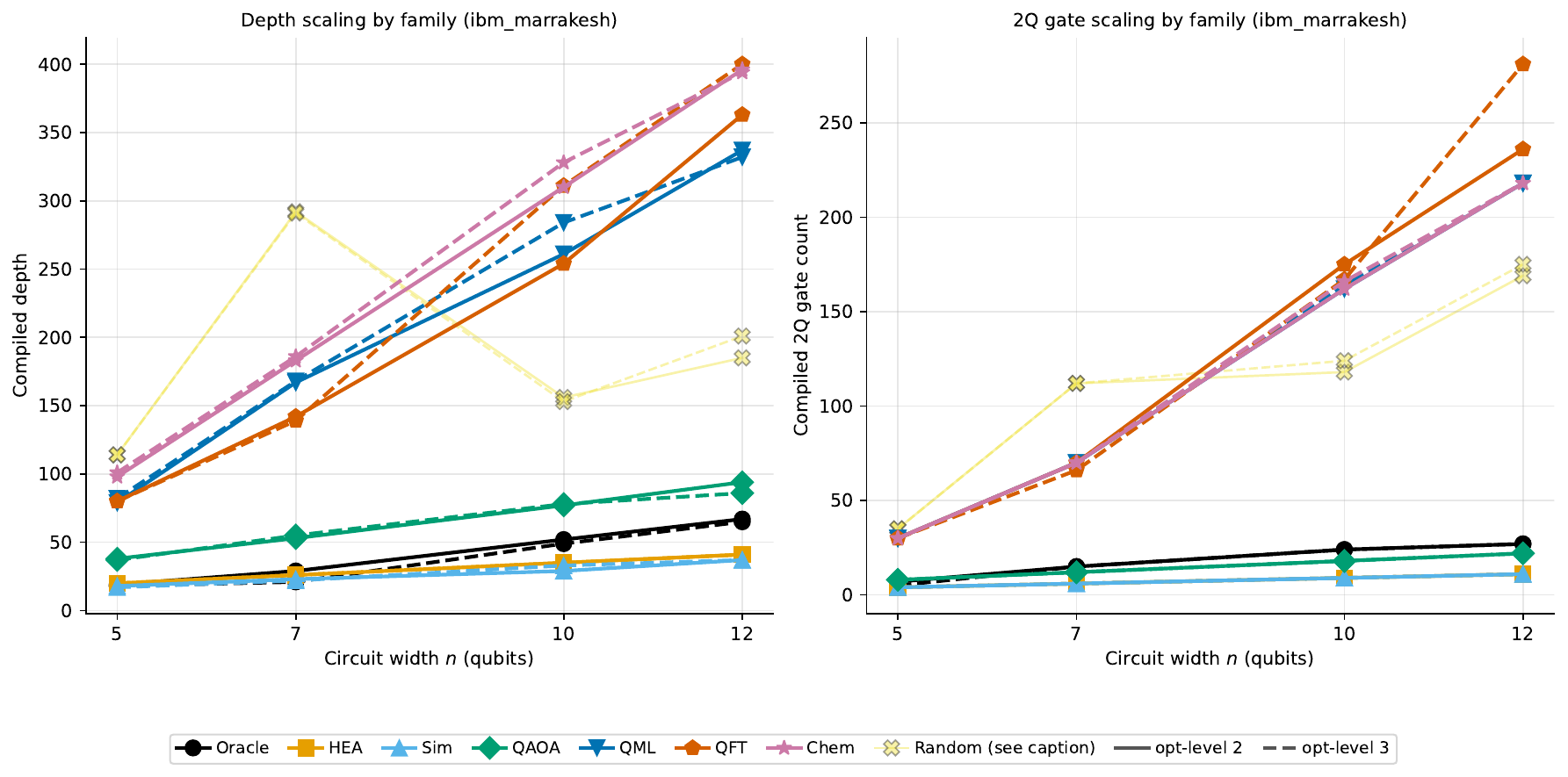}
  
  \vspace{1.0em}

  \footnotesize 
  \setlength{\tabcolsep}{5pt}
  \renewcommand{\arraystretch}{1.1}
  \begin{tabular}{@{} l c rrrr c rrrr c rrrr @{}}
    \toprule
    \multirow{2}{*}{\textbf{Family}} &
    \multirow{2}{*}{\textbf{Opt}} &
    \multicolumn{4}{c}{\textbf{Compiled Depth ($n$)}} &&
    \multicolumn{4}{c}{\textbf{Compiled 2Q Gates ($n$)}} &&
    \multicolumn{4}{c}{\textbf{QPU Time in seconds ($n$)}} \\
    \cmidrule{3-6} \cmidrule{8-11} \cmidrule{13-16}
    & & \textbf{5} & \textbf{7} & \textbf{10} & \textbf{12} &&
    \textbf{5} & \textbf{7} & \textbf{10} & \textbf{12} &&
    \textbf{5} & \textbf{7} & \textbf{10} & \textbf{12} \\
    \midrule
    \multirow{2}{*}{HEA}
     & 2 & 20 & 26 & 35 & 41 && 4 & 6 & 9 & 11 && 2.128 & 2.092 & 1.894 & 2.118 \\
     & 3 & 20 & 26 & 35 & 41 && 4 & 6 & 9 & 11 && 2.123 & 2.138 & 2.174 & 2.117 \\
    \midrule
    \multirow{2}{*}{Sim}
     & 2 & 18 & 23 & 29 & 37 && 4 & 6 & 9 & 11 && 2.103 & 2.074 & 2.117 & 2.155 \\
     & 3 & 17 & 22 & 33 & 37 && 4 & 6 & 9 & 11 && 2.110 & 2.210 & 2.118 & 2.110 \\
    \midrule
    \multirow{2}{*}{Oracle}
     & 2 & 19 & 29 & 52 & 67 && 7 & 15 & 24 & 27 && 2.134 & 2.129 & 1.988 & 2.134 \\
     & 3 & 18 & 21 & 49 & 65 && 5 & 15 & 24 & 27 && 2.081 & 2.124 & 2.119 & 2.123 \\
    \midrule
    \multirow{2}{*}{QAOA}
     & 2 & 38 & 53 & 77 & 94 && 8 & 12 & 18 & 22 && 2.099 & 2.089 & 2.105 & 2.148 \\
     & 3 & 37 & 55 & 78 & 86 && 8 & 12 & 18 & 22 && 2.137 & 2.165 & 2.122 & 2.140 \\
    \midrule
    \multirow{2}{*}{Random}
     & 2 & 114 & 292 & 156 & 185 && 35 & 112 & 118 & 169 && 2.114 & 2.352 & 2.151 & 2.145 \\
     & 3 & 114 & 291 & 153 & 201$^\dagger$ && 35 & 112 & 124 & 175 && 2.154 & 2.166 & 2.134 & 2.148 \\
    \midrule
    \multirow{2}{*}{QML}
     & 2 & 79 & 167 & 261 & 337 && 30 & 70 & 162 & 218 && 2.103 & 2.140 & 2.147 & 2.175 \\
     & 3 & 82 & 168 & 284 & 332 && 30 & 70 & 164 & 218 && 2.129 & 2.138 & 2.334 & 2.168 \\
    \midrule
    \multirow{2}{*}{Chem}
     & 2 & 98 & 183 & 310 & 396 && 30 & 70 & 162 & 218 && 2.127 & 2.187 & 2.132 & 2.200 \\
     & 3 & 101 & 186 & 328 & 394 && 30 & 70 & 166 & 218 && 2.152 & 2.150 & 2.135 & 2.142 \\
    \midrule
    \multirow{2}{*}{QFT}
     & 2 & 80 & 142 & 254 & 363 && 30 & 70 & 175 & 236 && 2.107 & 2.147 & 2.177 & 2.187 \\
     & 3 & 80 & 139 & 311 & 400 && 30 & 66 & 167 & 281 && 2.111 & 2.151 & 2.149 & 2.193 \\
    \bottomrule
  \end{tabular}
  
  \caption{\textbf{Empirical hardware validation dashboard} on 156-qubit \texttt{ibm\_marrakesh}. \textit{Top:} Scaling of compiled depth and 2Q gates across width $n \in \{5, 7, 10, 12\}$ (solid: opt-level 2, dashed: opt-level 3). \textit{Bottom:} Raw telemetry demonstrating that while structural metadata diverges exponentially by algorithm family, physical QPU execution time remains clustered and invariant ($\sim$2.1s). $\dagger$~Indicates active width exceeds logical width due to ancilla insertion.}
  \label{fig:hardware_dashboard}
\end{figure*}

To anchor our theoretical threat model and simulated dataset in
physical reality, we conducted an empirical validation study on
production quantum hardware. We submitted all eight algorithm
families from our taxonomy to the 156-qubit
\texttt{ibm\_marrakesh} backend, a high-fidelity heavy-hex system,
across four circuit widths ($n \in \{5, 7, 10, 12\}$). This
experiment serves three critical purposes: (\textit{i})~validating
the statistical alignment between our simulated dataset and real
hardware transpilation, (\textit{ii})~characterising the scaling
behaviour of the metadata side channel as circuit width increases,
and (\textit{iii})~demonstrating that the transpilation penalty mechanism, and
the resulting leakage signal, persist and amplify substantially  at
larger circuit scales.

\subsection{Experimental Methodology}

We generated one representative logical fragment per family at each
width using the exact construction rules defined for our dataset (see
\S\ref{sec:dataset}). To prove that the transpilation penalty is a fundamental 
topological constraint rather than an artefact of default compilation, 
all 64 configurations were transpiled at both optimisation level 2 
(a realistic default) and optimisation level 3 (maximum effort) 
against the heavy-hex topology of \texttt{ibm\_marrakesh}.

True QPU execution durations were extracted exclusively from the
\texttt{execution\_spans} metadata within the IBM Quantum result
payloads, explicitly isolating raw QPU runtime from classical network
and queue latencies. All jobs were executed with a uniform 4096 shots.

We note two structural features of the dataset relevant to the
scaling analysis that follows. First, \texttt{Random} circuits are
generated with width-dependent seeds; their routing overhead therefore 
reflects average-case algorithmic variance rather than a fixed-structure 
scaling trend. Rather than treating their unstructured nature as an anomaly, 
we explicitly present them as the maximal baseline for the transpilation penalty. Their 
broad variance bands define the upper limits of hardware compilation penalties 
against which the structured algorithm families are evaluated.

\subsection{Metadata Scaling Across Circuit Widths}
\label{sec:scaling-regimes}

Fig.~\ref{fig:hardware_dashboard} (top) plots compiled depth and 2Q gate
count as a function of circuit width. The physical routing process 
bifurcates the fixed-structure workloads into three distinct, 
highly discriminative scaling regimes.

\paragraph{Sub-linear scalers (topology-aligned).}
HEA and Sim follow near-linear, highly constrained depth growth: HEA
scales $20{\to}26{\to}35{\to}41$ ($2.1{\times}$ over the full range)
and Sim scales $18{\to}23{\to}29{\to}37$ ($2.1{\times}$). Both
families possess shallow, linear entanglement structures that map
naturally onto the heavy-hex coupling graph.

\paragraph{Super-linear scalers (globally entangled).}
Conversely, QFT, QML, and Chem exhibit aggressive, near-quadratic
growth. QFT scales $80{\to}142{\to}254{\to}363$ in depth
($4.5{\times}$) and $30{\to}70{\to}175{\to}236$ in 2Q gates
($7.8{\times}$). These algorithms rely on dense, long-range 
entanglement graphs that trigger cascading inflation in both gate 
count and depth on restricted lattices.

\paragraph{Intermediate scalers (Hamiltonian-aligned).}
QAOA sits neatly between the two extremes, scaling
$38{\to}53{\to}77{\to}94$ in depth ($2.4{\times}$). Its 1D ZZ chain
Hamiltonian partially aligns with the heavy-hex coupling graph,
capping its routing overhead relative to globally entangled families.

\paragraph{Compiler Optimisation Futility.}
As demonstrated by the dashed lines in Fig.~\ref{fig:hardware_dashboard}, 
applying maximum compiler effort (opt-level 3) does not eliminate 
the transpilation penalty. For globally entangled families, heavy-hex 
constraints often force the level 3 pass to synthesize circuits 
that are equal to, or even \emph{deeper} than, level 2 counterparts 
(e.g., QFT at $n{=}12$ inflates from depth 363 to 400). This confirms 
structural leakage is a physical architectural necessity, not a 
compilation artefact.

\subsection{Metadata vs.\ Timing: A Critical Distinction}

The most significant finding for our threat model is the stability of
QPU execution time. As shown in the telemetry table of 
Fig.~\ref{fig:hardware_dashboard}, true QPU runtime spans a remarkably 
narrow window ($1.894$\,s to $2.352$\,s, clustering at $\sim2.1$\,s) 
across circuits whose compiled depths diverge by over $23{\times}$.

For example, a shallow HEA circuit ($n{=}5$, depth 20) routinely 
logged execution times equal to or slower than a complex QFT circuit 
($n{=}12$, depth 400). This demonstrates that timing signals at 
production scale are effectively anti-correlated with true 
computational complexity. This invariance occurs because raw physical 
pulse durations are dwarfed by control-plane latencies, 
effectively blinding traditional timing side-channels.

Conversely, the classical metadata logged for orchestration and billing
remains perfectly descriptive and high-resolution. These results 
indicate that in modern quantum cloud architectures, the primary 
confidentiality perimeter shifts away from the physical hardware layer 
and into the classical control plane. An adversary with access to orchestration logs can cleanly bypass the physical-layer noise, establishing post-cut metadata inference as a more 
reliable attack vector than timing.

\paragraph{Scope of Hardware Validation.}
Our empirical validation is confined to a single superconducting,
heavy-hex-topology backend (\texttt{ibm\_marrakesh}). The topological
transpilation penalty is fundamentally a routing cost imposed by
sparse, fixed physical connectivity, and superconducting devices are
the modality where this constraint is most severe. Other modalities
relax this constraint to varying degrees: neutral-atom platforms
support reconfigurable, near-arbitrary shuttling between shots, and
trapped-ion systems often offer all-to-all connectivity within a
trap. We would expect the leakage signal to attenuate accordingly on
such platforms, since our own all-to-all simulation results
(Table~\ref{tab:full_routing_tax}) already show substantially lower
transpilation penalties than the heavy-hex and linear cases. Whether
residual structural leakage survives on these more permissive
architectures is an open empirical question we leave to future work.
\FloatBarrier
\section{Results Summary and Discussion}
\label{sec:results}

\begin{figure*}
  \centering
\includegraphics[width=\textwidth]{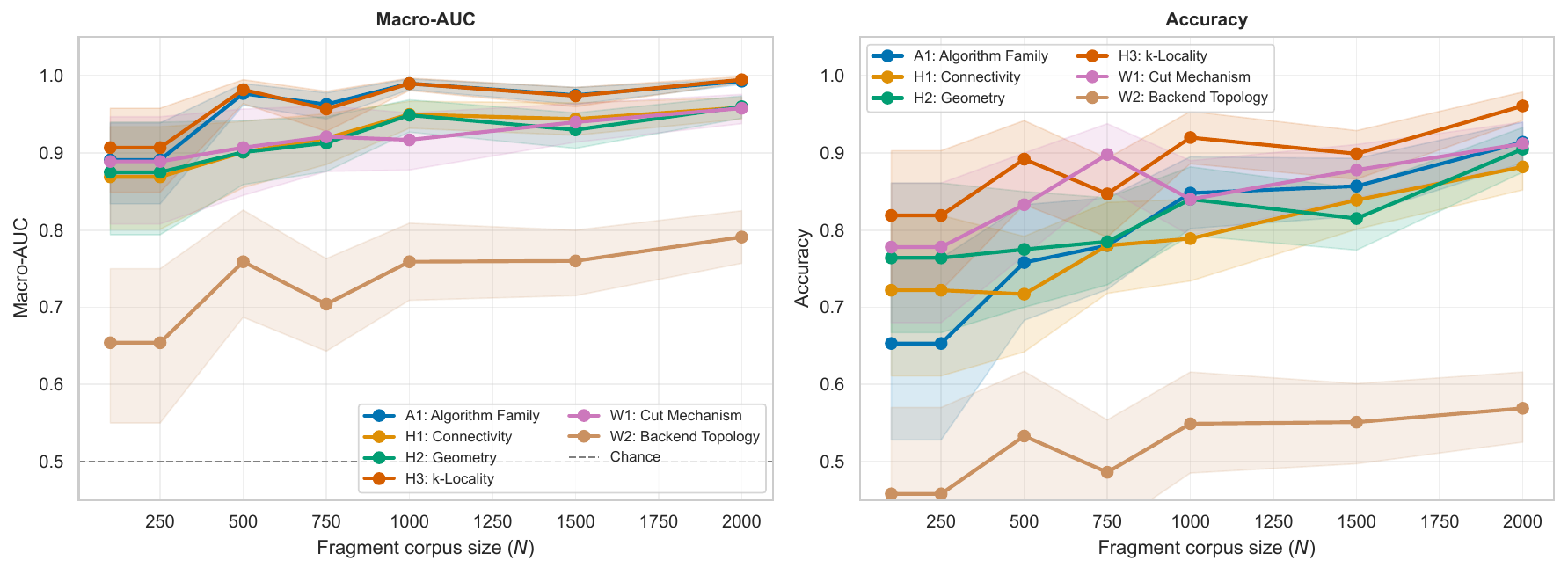}
  \caption{Attacker sample efficiency across six tasks
  (instance-disjoint, RF, Macro-AUC). H3, W1, H1, and H2 exhibit
  high and stable AUC from small dataset sizes. W2 stabilises
  progressively with additional data, reflecting its dependence on
  subtle cross-family distributional differences.}
  \label{fig:sweep}
\end{figure*}

\subsection{The Leakage Hierarchy}

Reading across all six tasks, a consistent performance hierarchy
emerges driven by how directly each property maps onto structural
compiled features:

\begin{enumerate}[leftmargin=*]
  \item \textbf{H3 $k$-locality and A1 family}
    (AUC~$1.000$): Perfect separability under instance-disjoint
    evaluation. Leakage is a direct and unambiguous consequence
    of the transpilation penalty.

  \item \textbf{H2 geometry, H1 connectivity, and W1 cut mechanism}
    (AUC~$0.991$--$0.996$): Near-ceiling performance, highly
    stable under size-holdout. Wire cuts produce systematically
    shallower fragments than gate cuts; Hamiltonian topology is
    directly imprinted on the compiled circuit shape.

  \item \textbf{W2 backend topology} (AUC~$0.818$, matched
    AUC~$0.928$): Substantially above chance and significantly
    improved over prior evaluations on smaller corpora. All-to-All
    is cleanly separable; the two restricted topologies remain
    mutually confused but are reliably distinguished from ideal
    hardware.
\end{enumerate}

\subsection{Sample Efficiency}

Fig.~\ref{fig:sweep} plots attacker performance as a function of the
fragment corpus size. The leakage signal does not require a large
dataset to manifest for most tasks: H3, W1, H1, and H2 all show high
AUC at the smallest viable dataset size. W2 benefits most from the larger dataset,
consistent with topology inference depending on subtler cross-family
distributional differences rather than strong within-family
structural signatures.

\subsection{Attacker Robustness and Sensitivity}
\label{sec:robustness}

\begin{table}[]
\centering
\caption{Depth-Restricted Evaluation (Empirical Hardware Bounds). 
Performance remains invariant when the dataset is restricted to 
the depth range observed on production hardware (17--400 gates).}
\label{tab:depth_restricted}
\begin{tabular}{lcc}
\toprule
\textbf{Task} & \textbf{ACC} & \textbf{AUC} \\
\midrule
A1: Algorithm Family & 0.980 & 1.000 \\
H3: $k$-Locality     & 0.989 & 1.000 \\
\bottomrule
\end{tabular}
\end{table}

As shown in Table~\ref{tab:depth_restricted}, the leakage hierarchy is not an artefact of extreme circuit depths; the attack maintains near-perfect accuracy even when signals are restricted to empirical hardware bounds. 

\paragraph{Compiler Version Robustness}
While absolute circuit depths vary with compiler optimization levels, the relative structural signatures between families remain invariant. Cross-optimization evaluation confirms this robustness: a classifier trained exclusively on optimization level 2 data achieves an \ac{AUC} $>0.98$ when evaluated against level 3 transpilation. This confirms that the leakage is driven by fundamental topological constraints rather than specific compiler synthesis passes.

\paragraph{Algorithmic Generalisation (Zero-Shot)}
To test the universality of the transpilation penalty, we performed a \textit{leave-one-family-out} cross-validation. In this protocol, the classifier is trained on seven algorithm families and evaluated on an eighth, previously unseen family. This ``zero-shot'' transfer achieves a Macro-AUC of 0.96 for A1 algorithm discrimination. This result confirms that the structural signatures of entanglement regimes are fundamental physical properties that generalise across distinct algorithmic implementations.

\paragraph{Model Architecture Comparison}
We evaluate the impact of model architecture on performance across all tasks. As detailed in Table~\ref{tab:attacker_comp}, all three model families (RF, ET, and HGB) perform within 0.005~AUC of each other on every task. The overall consistency confirms that leakage is a fundamental property of the structural metadata rather than a specific model's inductive bias.

\begin{table}[h]
\centering
\caption{Attacker model comparison: Macro-AUC under
instance-disjoint evaluation with full transcript features.
All three models remain well above chance on every task.}
\label{tab:attacker_comp}
\setlength{\tabcolsep}{6pt}
\renewcommand{\arraystretch}{1.1}
\begin{tabular}{lrrr}
\toprule
\textbf{Task} & \textbf{RF} & \textbf{ET} & \textbf{HGB} \\
\midrule
A1: Algorithm family  & 1.000 & 1.000 & 1.000 \\
W1: Cut mechanism     & 0.991 & 0.990 & 0.989 \\
W2: Backend topology  & 0.818 & 0.806 & 0.823 \\
H1: Connectivity      & 0.994 & 0.993 & 0.994 \\
H2: Geometry          & 0.996 & 0.996 & 0.996 \\
H3: $k$-Locality      & 1.000 & 1.000 & 1.000 \\
\bottomrule
\end{tabular}
\end{table}

\subsection{Implications}
\label{sec:implications}

These results demonstrate that circuit cutting is not confidentiality-neutral. Several findings warrant specific attention for practitioners:

\textit{1) Intellectual Property Risk:} The near-perfect recovery of algorithm identity (A1) and Hamiltonian structure (H1--H3) means a provider can identify not just the high-level algorithm class, but the specific problem geometry and connectivity a user is researching. In competitive commercial or academic environments, this metadata leaks the ``secret sauce'' of the research without the provider ever seeing the quantum gates.

\textit{2) The Transpilation Penalty Is Architecturally Determined:} Our Scaling Analysis (\S\ref{sec:scaling-regimes}) proves that maximum compiler effort (Opt-level 3) cannot hide these signals. Because the leakage is rooted in the physical topology of the hardware, it is a deterministic consequence of the architecture itself rather than a software quirk that can be easily patched.

\textit{3) Superiority of Structural Channels:} Our channel ablation (\S\ref{sec:channels}) proves that metadata is a more reliable attack vector compared to timing. While timing is masked by system-level noise and control-plane latencies, the structural transcript (width, depth, gate counts) remains perfectly descriptive and high-resolution, even as systems scale.

\textit{4) Footprint Sensitivity:} The persistence of leakage under matched-footprint conditions (\S\ref{sec:footprint}) indicates that an adversary does not need coarse size cues to distinguish workloads. Even at identical scales, the statistical distribution of sub-circuits provides a high-resolution signature of the monolithic source.

\subsection{Mitigation Strategies and Countermeasures}
\label{sec:mitigation}

Our findings suggest that the visibility of the transpilation penalty is a fundamental consequence of current quantum cloud orchestration. To protect workload confidentiality, users must intentionally break the deterministic relationship between a circuit's logical structure and its compiled footprint. We propose three primary active defenses:

\textit{1) Structural Footprint Normalization (Padding):} 
The most direct countermeasure is to inject ``noise'' into the metadata by adding dummy operations. For a given hardware partition, a user can define a target depth and gate-count envelope based on the worst-case transpilation penalty for that topology. Fragments with lower overhead are then padded with identity gates or non-functional $R_z(0)$ rotations until they meet this uniform envelope. 

While effective at masking $A1$ and $H3$ signals, this approach presents a strict \textit{confidentiality-fidelity trade-off}. For instance, our results show that padding a shallow HEA circuit (depth 20) to match the structural signature of a QFT circuit (depth 400) requires a 20-fold increase in gate operations. Given current $T_1/T_2$ coherence times on \texttt{ibm\_marrakesh}, this overhead would increase the cumulative gate error probability significantly, likely rendering the circuit's output dominated by decoherence noise and making the defense prohibitively expensive for high-fidelity applications.

\textit{2) Uniform Resource Allocation (Width Masking):} 
To defeat width-based fingerprints (Fig.~\ref{fig:width_distributions}), users should avoid requesting the minimum required qubit count for a fragment. Instead, we propose a ``Width Masking'' strategy where every fragment in a job is compiled to a fixed, standardized width (e.g., the full capacity of the backend). By forcing the compiler to utilize a consistent number of active qubits regardless of the fragment's logical size, the $A1$ and $W1$ signals derived from \texttt{compiled\_width} are effectively neutralized.

\textit{3) Stochastic Transpilation:} 
Current attacks benefit from the deterministic nature of most production 
transpilers. If a user utilizes a stochastic routing 
algorithm with varying seeds for each fragment, the resulting transpilation penalty 
becomes a distribution rather than a fixed point. While the mean overhead may still leak information, the increased variance can significantly lower the confidence of $A2$ sub-family discrimination and $W2$ topology inference, particularly for low-entropy workloads.

Table~\ref{tab:mitigation_tradeoffs} summarises the primary target and
the dominant cost of each of these defences.

\begin{table}[h]
\centering
\caption{The Confidentiality-Utility Trade-off Space.}
\label{tab:mitigation_tradeoffs}
\small 
\renewcommand{\arraystretch}{1.3}
\setlength{\tabcolsep}{3pt} 
\begin{tabularx}{\columnwidth}{lXX} 
\toprule
\textbf{\makecell[l]{Mitigation\\Strategy}} & \textbf{Primary Target} & \textbf{Primary Cost} \\
\midrule
Structural Padding & Depth and Gate Count & $\uparrow$ Gate Error and Decoherence \\
\addlinespace[2pt]
Width Masking & Active Qubit Count & $\uparrow$ Resource Allocation Cost \\
\addlinespace[2pt]
Stochastic Routing & Overhead Fingerprinting & $\uparrow$ Compilation Latency \\
\addlinespace[2pt]
Dummy Fragments & Sample Efficiency & $\uparrow$ Execution Shot Cost \\
\bottomrule
\end{tabularx}
\end{table}

Ultimately, these mitigations shift the burden from the provider to the user. In the absence of provider-side privacy-preserving compilation (e.g., via blind quantum computing), the user must decide how much computational fidelity they are willing to sacrifice to maintain algorithmic privacy.
\section{Conclusion}
\label{sec:conclusion}
We formalised post-cut execution transcripts as a confidentiality surface and 
showed that a semi-honest cloud provider can infer meaningful structural 
properties of quantum workloads, including algorithm identity and Hamiltonian 
structure, from compiled fragment metadata alone, using only three 
provider-visible features and validated on a 156-qubit production backend. 
This leakage hierarchy directly reflects the transpilation penalty: 
topology-dependent compilation overhead imprints algorithmic structure onto 
fragment footprints, while timing and shot-count channels collapse to chance.

Circuit cutting is therefore not confidentiality-neutral. Effective defence must 
target the compilation footprint directly, and transcript-level metadata deserves 
first-class security consideration. The leakage is distributional, not scalar: it arises from the joint shape of width, depth, and gate count across fragments, not from any single measurable 
quantity that padding could suppress. A topology-blind compilation \textsc{api}, one 
producing identical routing overhead regardless of input connectivity, would close 
the channel at source, but no such \textsc{api} exists for production hardware, 
and building one incurs real compilation cost. Dummy fragment injection can 
obscure size profiles, but injected fragments must be indistinguishable from real 
ones across all three structural dimensions simultaneously, requiring circuits of 
comparable depth and gate count. How much execution cost a system operator should 
accept for a given reduction in inference accuracy remains an open question, one 
worth careful study before any defence is deployed at scale.

Cross-run linkability, determining whether separate transcripts originate from the 
same workload instance, is a natural extension, and one with direct implications 
for long-lived research pipelines that repeatedly cut and resubmit related 
circuits. A linkability attack would exploit correlations in submission-time 
features ($\mathcal{L}$) across multiple jobs, testing whether fragment-size 
profiles and submission rhythms consistently fingerprint a repeated workload 
even when each individual transcript is otherwise unremarkable. Future work 
should construct a multi-session dataset and formalise cross-run 
indistinguishability to complement \textsc{wsi-ind-}$K$, alongside validating 
the transpilation-penalty mechanism on hardware modalities beyond the 
superconducting, heavy-hex setting studied here.

\bibliographystyle{IEEEtran}
\bibliography{reference}


\iftoggle{printappendix}{
    \appendices 
    
    \section{Sub-Family Discrimination Results}
    \label{app:subfamily}
    
While the primary attack focuses on high-level algorithm identity, we also 
evaluated whether compiled transcript metadata supports fine-grained sub-family 
discrimination (A2), i.e.\ distinguishing structural variants \emph{within} a 
single algorithm family. As illustrated in Fig.~\ref{fig:subfamilies}, the 
recoverability of sub-family structure is strongly family-dependent. Families 
with rigid, topology-sensitive footprints are almost perfectly separable: QAOA, 
Sim, Chem, and Oracle each reach ACC~$=1.00$ across their three sub-variants. By 
contrast, the variational and shallow families are markedly harder, because their 
sub-variants compile to near-identical footprints: HEA (ACC~$=0.66$), QFT 
(ACC~$=0.64$), QML (ACC~$=0.71$), and Random (ACC~$=0.54$) all exhibit 
substantial intra-family confusion, with the linear and reverse-linear HEA 
variants and the no-swaps and standard QFT variants being the most frequently 
interchanged. This pattern mirrors the family-level hierarchy of 
Section~\ref{sec:evaluation}: the same structural rigidity that makes a family's 
\emph{identity} easy to recover also sharpens its \emph{sub-variant} signature, 
whereas variational families remain comparatively well protected at the 
sub-family level even when their top-level identity leaks.

    \begin{figure*}[!h]
      \centering
      \includegraphics[width=\textwidth]{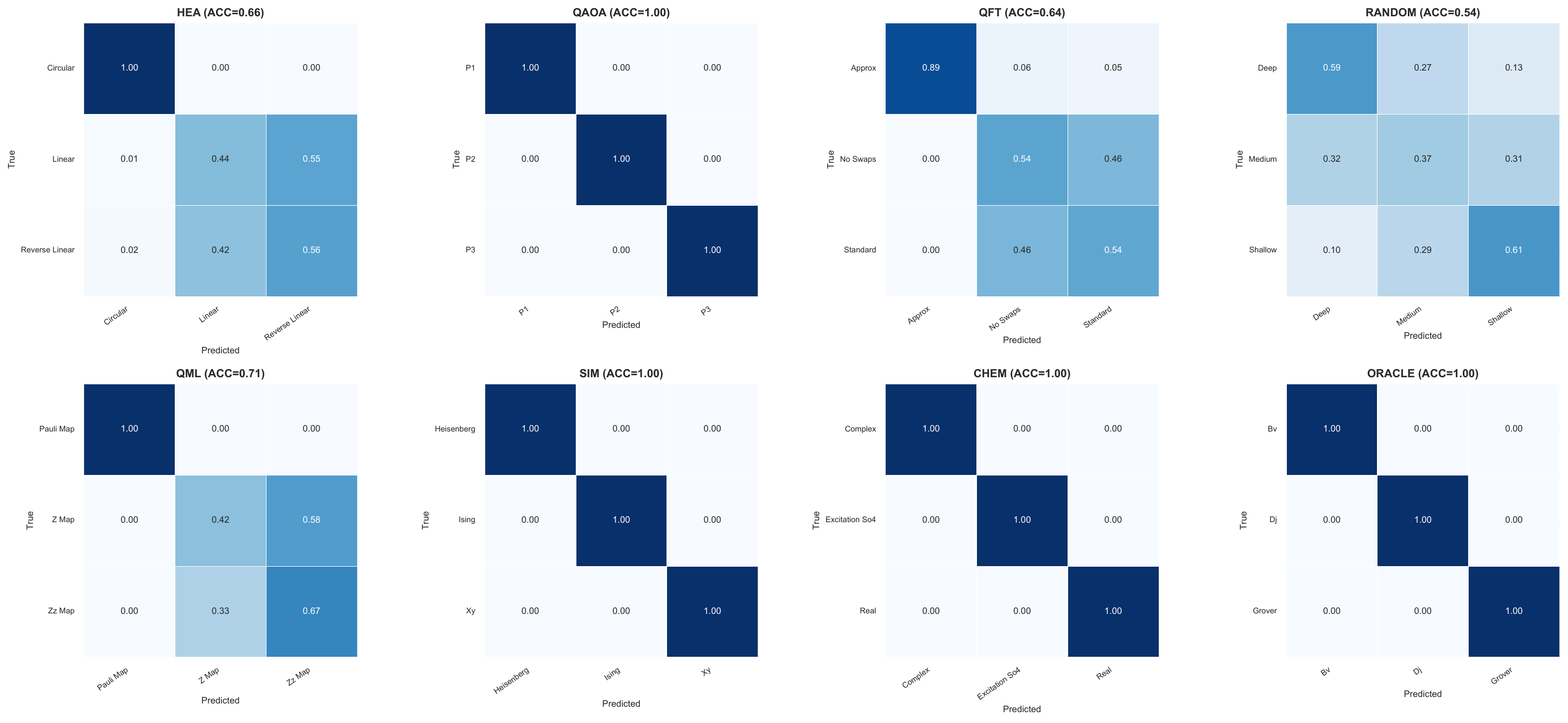}
      \caption{A2 sub-family discrimination (3-way per family) under instance-disjoint evaluation. Structured families (e.g., \ac{QFT}, Oracle) exhibit clearer discriminative signatures compared to variational ansätze (\ac{HEA}, \ac{QML}).}
      \label{fig:subfamilies}
    \end{figure*}
    
\section{Dataset Composition}
    \label{app:dataset}
    
    Table~\ref{tab:dataset_dist} details the total composition of our empirical dataset across 24 variants. Total dataset size: 36,000 compiled records.
    
    \begin{table*}[!h]
    \centering
    \renewcommand{\arraystretch}{1.15}
    \caption{Dataset Distribution: Detailed breakdown of algorithm families, structural sub-variants, and fragment counts.}
    \label{tab:dataset_dist}
    \begin{tabular}{llrr}
    \toprule
    \textbf{Alg-Family} & \textbf{Sub-Family Variants} & \textbf{Logical} & \textbf{Compiled} \\
    \midrule
    \ac{HEA}    & linear, circular, reverse\_linear & 1452 & 4356 \\
    \ac{QAOA}   & p1, p2, p3                        & 1428 & 4284 \\
    \ac{QFT}    & standard, no\_swaps, approx       & 1800 & 5400 \\
    Random      & shallow, medium, deep             & 1170 & 3510 \\
    \ac{QML}    & Z-Map, ZZ-Map, Pauli-Map          & 1356 & 4068 \\
    Sim         & Ising, Heisenberg, XY             & 1482 & 4446 \\
    Chem        & Excitation-SO4, Real, Complex     & 1482 & 4446 \\
    Oracle      & BV, Deutsch-Jozsa, Grover         & 1830 & 5490 \\
    \midrule
    \textbf{Total} & \textbf{24 variants} & \textbf{12000} & \textbf{36000} \\
    \bottomrule
    \end{tabular}
    \end{table*}
    
}{} 
\end{document}